\documentclass{article}

 \usepackage[preprint]{neurips_2026}

\usepackage[utf8]{inputenc} 
\usepackage[T1]{fontenc}    
\usepackage{hyperref}       
\usepackage{url}            
\usepackage{booktabs}       
\usepackage{amsfonts}       
\usepackage{nicefrac}       
\usepackage{microtype}      
\usepackage{xcolor}         
\usepackage{graphicx}
\usepackage{placeins}
\usepackage{amsmath,amssymb,amsfonts,amsthm,mathtools}
\newtheorem{theorem}{Theorem}
\newtheorem{proposition}{Proposition}
\usepackage{graphicx}
\usepackage{subcaption}
\usepackage{multirow}

\usepackage{algorithm}
\usepackage{algorithmic}

\usepackage[capitalize,noabbrev]{cleveref}

\newcommand{\corrauthormark}{\textsuperscript{*}}
\newcommand{\corrauthornote}{
    \begingroup
    \renewcommand{\thefootnote}{*}
    \footnotetext{Corresponding authors: Zhenzhong Wang and Min Jiang,
    zhenzhongwang@xmu.edu.cn, minjiang@xmu.edu.cn.}
    \endgroup
}

\title{Geometry-Aware Anisotropic Boundary Correction for Aerodynamic Simulation}

\author{
    Xin Zhang$^{1}$,
    Yipeng Huang$^{1}$,
    Shu Jiang$^{2}$,
    Zhenzhong Wang$^{1}$\corrauthormark,
    Min Jiang$^{1}$\corrauthormark\\
    $^{1}$School of Informatics, Xiamen University, Xiamen, China\\
    $^{2}$Institute of Artificial Intelligence, Xiamen University, Xiamen, China
}

\begin{document}

\maketitle
\corrauthornote

\begin{abstract}
Aerodynamic simulation is a key component of engineering shape design, where core quantities such as the surface pressure coefficient strongly depend on flow dynamics near solid boundaries. Neural operators provide an efficient alternative to expensive Computational Fluid Dynamics (CFD) solvers. However, conventional methods treat the boundary region isotropically, failing to account for the distinct physical behaviors along the boundaries. In reality, the aerodynamic process exhibits anisotropy: along the tangential direction, flow propagates along the wall; along the normal direction, physical quantities are constrained by the wall. 
To explicitly model the distinct physical behaviors, we propose GeoABC, a geometry-conditioned anisotropic boundary correction framework. GeoABC leverages the boundary geometries to introduce direction-aware boundary correction into the intermediate representations of neural operators, transforming boundary geometry from static input features into a structural prior that modulates physical prediction. On 2D airfoil and 3D car tasks, GeoABC consistently adapts to multiple neural operator backbones, reducing near-boundary relative $L_2$ error by $\sim$38\% on average, narrowing the structural near-wall gap shared by mainstream neural operators, and advancing neural operators toward high-fidelity aerodynamic simulation.

\end{abstract}

\section{Introduction}
\label{sec:intro}

Steady-state aerodynamic simulation plays a central role in engineering design. The design of cars, wings, and other geometries all relies on Computational Fluid Dynamics (CFD) to evaluate key physical quantities such as pressure fields, velocity fields, and moments~\citep{anderson1995cfd,jameson1988aerodynamic,blazek2015cfd}. Classical CFD solvers can provide high-fidelity solutions. However, due to high computational cost, CFD struggles to be integrated into large-scale design optimization, rapid shape screening, and interactive engineering iteration~\citep{ferziger2002cfd, martins2022aerodynamic, green1997method}.

\begin{figure}[t]
	\centering
    \includegraphics[width=1.0\linewidth]{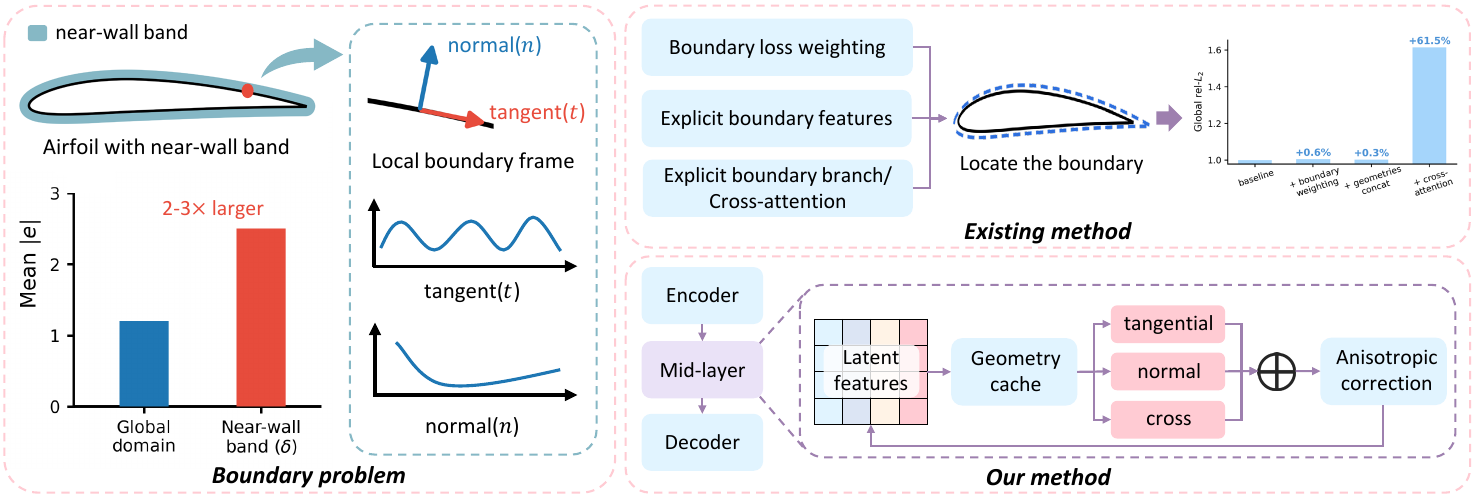}
    \caption{GeoABC addresses the near-wall accuracy gap left by existing boundary-aware remedies by using boundary geometry as a structural prior for direction-aware correction.}
	\label{fig:intro}
\end{figure}

In recent years, a wide range of deep learning methods have been developed as efficient alternatives to traditional numerical solvers for PDE-governed physical simulation~\citep{karniadakis2021pinn,raissi2019pinns,brunton2020machine}. Among these, neural operators learn infinite-dimensional function space mappings from geometry, boundary conditions, or system parameters to the physical solution fields~\citep{lu2021deeponet,li2021fno,kovachki2023nokernel}. Once trained, they can predict pressure fields, velocity fields, and other solution fields through a single forward pass, thereby greatly accelerating steady-state aerodynamic simulation~\citep{catalani2024neural, jiang2025prediction, catalani2025towards}. 

Despite their significant progress in global flow-field prediction, aerodynamic surrogates still face a unique and critical challenge: physical fields in the near-wall region, namely the narrow boundary neighborhood adjacent to solid aerodynamic surfaces, are difficult to predict accurately. On one hand, near-wall regions involve complex mechanisms such as wall constraints, pressure gradients, and local flow separation. These mechanisms lead to strong spatial heterogeneity and local high-frequency variations~\citep{tanarro2020effect}. Compared with relatively smooth flow structures in the far field, near-wall velocity and pressure often change rapidly along the wall-normal direction, while exhibiting different propagation behaviors along the tangential direction. This makes it difficult for neural operators to maintain near-wall accuracy comparable to their global-field accuracy~\citep{vinuesa2017pressure}. On the other hand, near-wall errors do not only affect local pressure or velocity values. They also directly enter the evaluation of aerodynamic performance. Taking an airfoil or a vehicle surface as an example, the wall pressure distribution can be normalized into the surface pressure coefficient, $C_p$. The shape of the $C_p$ curve is closely related to key flow phenomena such as suction peaks, pressure recovery, and potential separation. The same wall pressure distribution also determines surface pressure loads, which further affect the estimation of lift, drag, and moments~\citep{spalart2000strategies, badrya2021unsteady,pomeroy2017design}. Therefore, even though the near-wall region occupies only a small portion of the computational domain, local errors in this region may propagate into biased aerodynamic performance metrics, reducing the reliability of surrogate models in engineering design~\citep{brunton2020machine, vinuesa2022enhancing}.

To alleviate this problem, existing methods typically enhance boundary awareness by weighting boundary losses, explicitly concatenating geometry descriptors such as Signed Distance Function or curvature, or introducing additional boundary branches. These geometric cues are derived from the geometries already available in CFD tasks, such as airfoil profiles, car surfaces, or solid boundary meshes. However, most of these methods still treat the boundary as a special location that should receive more attention, rather than a physical region with an inherent directional structure~\citep{li2023geofno,hao2023gnot,li2023gino}. In contrast, classical CFD often uses body-fitted meshes or wall-normal stretched meshes near-wall regions, increasing resolution along the wall-normal direction to resolve wall constraints, pressure gradients, and near-wall variations~\citep{ferziger2002cfd, thompson1985numerical}. This indicates that physical variations near aerodynamic boundaries are not isotropic but naturally depend on the local tangential-normal coordinate structure. Therefore, simply using geometry descriptors as input features can only inform the model which points are near the boundary, but fail to capture how physical quantities propagate tangentially along the boundary, change rapidly along the wall normal, and how these two directional behaviors interact. As illustrated in Fig.~\ref{fig:intro}, such boundary representations are not well aligned with the geometry-physics anisotropic coupling required for accurate boundary prediction in aerodynamics.

To bridge this gap, we propose GeoABC, a geometry-conditioned anisotropic boundary correction framework. Instead of using geometric descriptors merely as static input features, GeoABC leverages the geometries already available in CFD tasks to shift near-wall prediction from isotropic proximity modeling to anisotropic geometry-physics coupling. Specifically, GeoABC explicitly constructs local tangential-normal structures from these geometries, allowing the model to distinguish boundary-aligned propagation from wall-normal gradient variations. Meanwhile, it performs direction-aware boundary correction in the intermediate representation space of neural operators. This allows geometry to improve the physical predictions of different backbones in a plug-and-play manner, rather than being used only as static features concatenated at the input stage. 
Thus, the contributions of this paper are summarized as follows:
\begin{itemize}
\item we systematically diagnose near-wall errors in steady-state aerodynamic surrogates. We show that these errors cannot be simply attributed to data quality, training strategy, or model capacity, but are closely related to the isotropic representation of boundary physics in existing neural operators.
\item we propose GeoABC, a geometry-conditioned, frame-aware, and plug-and-play anisotropic boundary correction framework that transforms the directional boundary prior implicit in classical CFD into a representation correction mechanism for neural operators.
\item We validate the method on 2D airfoil and 3D car tasks, showing that GeoABC consistently reduces near-wall prediction errors across multiple mainstream backbones and improves engineering metrics that depend directly on the boundary.
\end{itemize}








\section{Related Work}
\label{sec:related}

Neural operators have become an important paradigm for aerodynamic simulation, with the core goal of learning the solution operator underlying PDE-governed systems, i.e., the mapping from input functions, geometric conditions, or physical parameters to solution fields~\citep{kovachki2023neural, diab2025temporal, calvello2025continuum, lin2026fano}. Early representative methods such as DeepONet~\citep{lu2021deeponet} and FNO~\citep{li2021fno} model function-space mappings through branch--trunk decomposition and spectral integral operators, respectively. However, these early formulations are limited in geometry-dependent aerodynamic prediction, as they are not naturally suited to irregular geometries, varying computational domains, or non-uniform samples induced by complex shapes. Later Transformer-based operators, such as GNOT~\citep{hao2023gnot}, Transolver~\citep{wu2024transolver}, and DPOT~\citep{hao2024dpot}, further improve the modeling capacity of neural operators for complex physical fields, irregular samples, and large-scale geometries. Nevertheless, these advances mainly improve global field prediction on complex geometries, while substantial near-wall errors still remain.

To further enhance the awareness of complex geometries and local boundary structures, recent works explicitly incorporate geometry, boundary, or interface information~\citep{alkin2024universal, rahman2024pretraining, herde2024poseidon}. Geo-FNO~\citep{li2023geofno} and GINO~\citep{li2023gino} improve the adaptability to irregular geometries through geometric mappings or geometry encodings. GeoPT~\citep{wu2026geoptscalingphysicssimulation} and From Cheap Geometry to Expensive Physics~\citep{zhang2026cheap} leverage low-cost geometric pre-training to reduce the dependence on expensive physics labels. BENO~\citep{wang2024beno} and IANO~\citep{wang2026cross} further show that boundary or interface information can provide useful structural cues for neural PDE modeling. Overall, these methods demonstrate the importance of geometric information, but they mainly use geometry as a computational-domain representation, input condition, pre-training supervision, or boundary cue, without treating geometric solid boundaries as structural priors that organize near-wall physical variations~\citep{alkin2025ab,zeng2025point}. In this work, we instead use geometry as a structural prior to organize anisotropic near-wall responses, advancing neural operators toward boundary-accurate aerodynamic simulation.

\section{Methodology}
\label{sec:method}

\subsection{Problem Setup}
\label{sec:method:problem}

We study steady-state aerodynamic operator learning on geometry-dependent domains.
For a geometry $G\in\mathcal{G}$, let $\Omega_G\subset\mathbb{R}^d$ denote the fluid domain and $\partial\Omega_G$ its solid boundary.
Given a query point $\mathbf{x}\in\overline{\Omega}_G$ and optional operating conditions $\mathbf{s}$, a neural operator predicts the physical state
\begin{equation}
    \mathcal{F}_{\theta}:(\mathbf{x},G,\mathbf{s})\mapsto \mathbf{u}(\mathbf{x}),
\end{equation}
where $\mathbf{u}(\mathbf{x})$ denotes the target field, such as velocity and pressure.
When the operating conditions are fixed, we omit $\mathbf{s}$ and write $\mathcal{F}_{\theta}:(\mathbf{x},G)\mapsto \mathbf{u}(\mathbf{x})$.

We define the near-wall band as
\begin{equation}
    \Omega_{\delta}(G)
    =
    \{\,\mathbf{x}\in\Omega_G:\operatorname{dist}(\mathbf{x},\partial\Omega_G)<\delta\,\}.
\end{equation}
Instead of using a hard near-wall mask, GeoABC uses a continuous geometry-confidence gate to modulate the strength of boundary correction.
The correction is weakened far from the boundary or where local geometric descriptors are unreliable, and is strengthened in reliable near-wall regions.

\subsection{Overview}
\label{sec:method:overview}

GeoABC is a geometry-conditioned boundary correction module inserted into an intermediate layer of a backbone neural operator.
As shown in Fig.~\ref{fig:geoabc_pipeline}, GeoABC takes two inputs for each query point: the backbone latent representation and the local geometric descriptor from an offline geometry cache.
The geometry cache provides the local tangent--normal frame, curvature, near-wall window, and spatial confidence, while the latent representation provides the physical context to be corrected.


Concretely, the boundary correction proceeds as follows.
First, GeoABC projects the intermediate representation into a low-dimensional anchor space, which serves as a physical interface for directional decomposition, such as velocity direction.
Second, the anchor is decomposed in the local tangent--normal frame induced by the solid boundary.
Third, GeoABC generates tangential, normal, and tangent--normal coupled boundary corrections with geometry-conditioned direction gates.
Finally, the correction is lifted back to the backbone latent space and written into the intermediate representation under a geometry-confidence gate.
The corrected representation is then propagated through the backbone and decoder to produce the final physical prediction.

\begin{figure}[t]
    \centering
    \includegraphics[width=\linewidth]{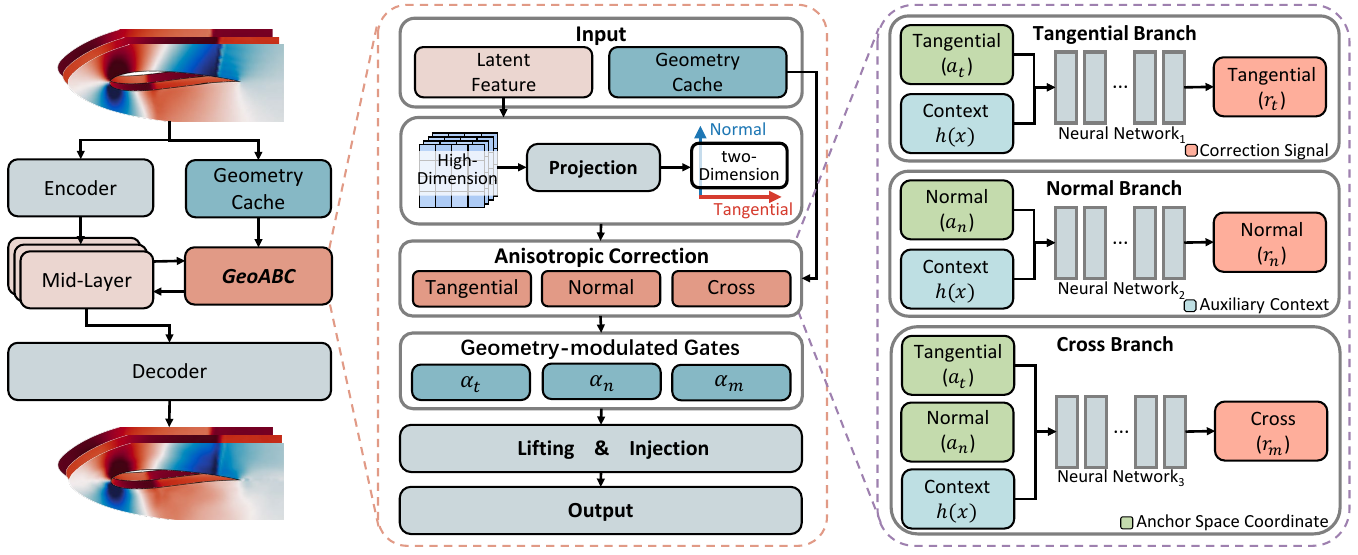}
    \caption{
    Overview of GeoABC.
GeoABC inserts a geometry-conditioned anisotropic correction module into the intermediate layer of a backbone neural operator.
It uses the geometry cache to construct a tangent-normal coordinate frame and projects the backbone latent feature into this local frame.
The module then generates correction signals, modulates them with local geometric conditions, and writes them back to the backbone intermediate representation for near-wall aerodynamic prediction.
    }
    \label{fig:geoabc_pipeline}
\end{figure}

\subsection{Geometry-Conditioned Anisotropic Boundary Correction}
\label{sec:method:correction}

\paragraph{Local geometric descriptor and spatial confidence.}
Given a geometry $G$ and a query point $\mathbf{x}$, GeoABC first constructs a local geometric descriptor for each query point:
\begin{equation}
    \mathbf{g}_G(\mathbf{x})
    =
    \big(
    \phi(\mathbf{x}),\,
    \hat{\mathbf{t}}(\mathbf{x}),\,
    \hat{\mathbf{n}}(\mathbf{x}),\,
    \kappa(\mathbf{x}),\,
    w(\mathbf{x}),\,
    c(\mathbf{x})
    \big).
    \label{eq:geo_descriptor}
\end{equation}
Here, $\phi(\mathbf{x})$ denotes the signed distance from the query point to the solid boundary, $\hat{\mathbf{t}}(\mathbf{x})$ and $\hat{\mathbf{n}}(\mathbf{x})$ are the local unit tangent and normal vectors, respectively.
The term $\kappa(\mathbf{x})$ denotes the local curvature.
The weight $w(\mathbf{x})\in[0,1]$ is a smooth narrow-band window that reflects the proximity of $\mathbf{x}$ to the boundary region.
The term $c(\mathbf{x})\in[0,1]$ measures the reliability of the local geometric descriptor.


To make the correction mainly act on reliable near-wall regions, GeoABC defines a spatial confidence:
\begin{equation}
    \omega_G(\mathbf{x})
    =
    w(\mathbf{x})\cdot c(\mathbf{x}).
    \label{eq:spatial_confidence}
\end{equation}
Here, $w(\mathbf{x})$ smoothly controls the spatial strength of the boundary correction according to the distance from $\mathbf{x}$ to the boundary, while $c(\mathbf{x})$ measures the reliability of the local geometric description.
Therefore, when the query point is close to the solid boundary and its geometric descriptor is reliable, $\omega_G(\mathbf{x})$ takes a relatively large value, and GeoABC injects a stronger direction-aware correction into the intermediate representation at this location.
As the query point moves away from the boundary, $w(\mathbf{x})$ gradually decreases, so $\omega_G(\mathbf{x})$ smoothly approaches $0$.
This reduces the influence of GeoABC on the far-field or bulk-region representations.
In practice, all geometric descriptors can be precomputed offline and stored as a geometry cache.

\paragraph{Anchor projection and local-frame decomposition.}
After obtaining the local geometric frame, GeoABC needs to convert the intermediate representation of the backbone into a space where directional decomposition can be performed.
Let $\mathbf{z}_{\ell}(\mathbf{x})\in\mathbb{R}^{d_z}$ denote the backbone's intermediate representation at layer $\ell$ and query point $\mathbf{x}$.
Although $\mathbf{z}_{\ell}(\mathbf{x})$ already contains physical context extracted by the backbone from the global input, its channels do not have explicit geometric directions.
Therefore, we cannot directly identify tangential or normal components from these latent channels.

To address this issue, GeoABC first maps the high-dimensional intermediate representation to a two-dimensional anchor space with physical directional meaning through a pointwise linear down-projection $\mathbf{W}_{\downarrow}$:
\begin{equation}
    \mathbf{a}(\mathbf{x})
    =
    \mathbf{W}_{\downarrow} \mathbf{z}_{\ell}(\mathbf{x})
    \in \mathbb{R}^{2}.
    \label{eq:anchor_projection}
\end{equation}
In addition, GeoABC uses an anchor alignment loss $\mathcal{L}_{\mathrm{anchor}}$ (Eq.~\ref{eq:anchor_loss}) to constrain $\mathbf{a}(\mathbf{x})$ and align it with physical directional quantities.
For example, in the $(u,v,p)$ setting, $\mathbf{a}(\mathbf{x})$ can be aligned with the velocity vector $(u,v)$.
In the pressure-only setting, $\mathbf{a}(\mathbf{x})$ can be aligned with the pressure-gradient direction $\nabla p$.
In this way, the anchor space is no longer an arbitrary latent space, but rather a low-dimensional space with a clear physical directional meaning.

GeoABC then uses the local tangent and normal vectors from the geometry cache to construct a solid-boundary frame, providing a structural prior for the anchor representation:
\begin{equation}
    \mathbf{B}_G(\mathbf{x})
    =
    \big[
    \hat{\mathbf{t}}(\mathbf{x}),\,
    \hat{\mathbf{n}}(\mathbf{x})
    \big]
    \in \mathbb{R}^{2\times 2}.
    \label{eq:local_frame}
\end{equation}
The two columns of $\mathbf{B}_G(\mathbf{x})$ correspond to the local tangential and normal directions, respectively.
Projecting the anchor $\mathbf{a}(\mathbf{x})$ onto this local frame gives
\begin{equation}
    \mathbf{a}_B(\mathbf{x})
    =
    \mathbf{B}_G(\mathbf{x})^{\top} \mathbf{a}(\mathbf{x})
    =
    \begin{bmatrix}
    a_t(\mathbf{x})\\
    a_n(\mathbf{x})
    \end{bmatrix}
    \in \mathbb{R}^{2},
    \label{eq:anchor_frame_projection}
\end{equation}
where $a_t(\mathbf{x})$ and $a_n(\mathbf{x})$ denote the tangential and wall-normal coordinates of the anchor.
It is worth noting that they are not a tangential-normal decomposition of the final physical field.
Instead, they serve as local-coordinate inputs for the subsequent direction-aware correction.

\paragraph{Direction-aware correction branches and geometric modulation.}
After obtaining the local coordinates $(a_t(\mathbf{x}),a_n(\mathbf{x}))$, GeoABC generates direction-aware correction signals in the boundary-aligned $(t,n)$ frame.
Specifically, we first extract an auxiliary context from the same intermediate representation $\mathbf{z}_{\ell}(\mathbf{x})$:
\begin{equation}
    \mathbf{h}(\mathbf{x})
    =
    \mathbf{W}_h \mathbf{z}_{\ell}(\mathbf{x}).
    \label{eq:aux_context}
\end{equation}
Here, $\mathbf{h}(\mathbf{x})$ preserves the local and global physical information already encoded by the backbone.
GeoABC then uses three pointwise MLP branches to generate correction signals along the tangential, normal, and cross directions:
\begin{align}
    r_t(\mathbf{x})
    &=
    f_t\big(a_t(\mathbf{x}),\mathbf{h}(\mathbf{x})\big),
    \label{eq:tangent_branch}
    \\
    r_n(\mathbf{x})
    &=
    f_n\big(a_n(\mathbf{x}),\mathbf{h}(\mathbf{x})\big),
    \label{eq:normal_branch}
    \\
    r_{m}(\mathbf{x})
    &=
    f_{m}\big(a_t(\mathbf{x}),a_n(\mathbf{x}),\mathbf{h}(\mathbf{x})\big).
    \label{eq:cross_branch}
\end{align}
The tangential branch $f_t$ generates a scalar correction signal along the boundary, the normal branch $f_n$ generates a scalar correction signal normal to the wall, and the cross branch $f_m$ models the coupled correction between the tangential and normal directions.

Since the local geometric conditions vary across boundary locations, the importance of different directional corrections should also change from point to point.
Therefore, GeoABC further uses a lightweight geometric modulator to predict three gating coefficients from the local geometric descriptor $\mathbf{g}_G(\mathbf{x})$:
\begin{equation}
    \big(
    \alpha_t(\mathbf{x}),\,
    \alpha_n(\mathbf{x}),\,
    \alpha_{m}(\mathbf{x})
    \big)
    =
    \sigma
    \left(
    \Gamma_{\psi}\big(\mathbf{g}_G(\mathbf{x})\big)
    \right).
    \label{eq:directional_gates}
\end{equation}
Here, $\Gamma_{\psi}$ is the geometric modulator, and $\sigma(\cdot)$ denotes the sigmoid function.
The three gating coefficients control the strengths of the tangential, normal, and cross corrections, respectively.
In other words, the correction at the same near-wall point is not fixed.
Instead, GeoABC adaptively adjusts the contribution of each directional correction according to the local geometric condition.

Since the cross branch produces a coupled correction signal, GeoABC further maps it to a two-dimensional correction in the local frame space:
\begin{equation}
    \mathbf{q}_{m}(\mathbf{x})
    =
    \mathbf{W}_{m} r_{m}(\mathbf{x})
    \in \mathbb{R}^{2}.
    \label{eq:cross_lift}
\end{equation}
where $\mathbf{W}_m$ is a learnable linear lifting map.
It converts the output of the cross branch into a two-dimensional correction vector in the local $(t,n)$ frame.
As a result, the cross correction can be composed with the tangential and normal corrections in the same local coordinate system.

\paragraph{Directional correction synthesis and mid-layer write-back.}
After obtaining the three direction-aware correction signals, GeoABC first composes the overall boundary correction in the local tangent-normal frame:
\begin{equation}
    \mathbf{c}_B(\mathbf{x})
    =
    \begin{bmatrix}
        \alpha_t(\mathbf{x}) r_t(\mathbf{x})
        \\
        \alpha_n(\mathbf{x}) r_n(\mathbf{x})
    \end{bmatrix}
    +
    \alpha_{m}(\mathbf{x}) \mathbf{q}_{m}(\mathbf{x})
    \in \mathbb{R}^{2}.
    \label{eq:frame_correction}
\end{equation}
The first term represents the geometrically gated tangential and normal corrections.
The second term represents the coupled correction between the tangential and normal directions.
Since $\mathbf{c}_B(\mathbf{x})$ is defined in the local tangent-normal coordinate system, GeoABC then maps it back to the global two-dimensional anchor space using the local frame $\mathbf{B}_G(\mathbf{x})$:
\begin{equation}
    \mathbf{c}_{\mathrm{amb}}(\mathbf{x})
    =
    \mathbf{B}_G(\mathbf{x})\mathbf{c}_B(\mathbf{x})
    \in \mathbb{R}^{2}.
    \label{eq:ambient_correction}
\end{equation}
The vector $\mathbf{c}_{\mathrm{amb}}(\mathbf{x})$ can be viewed as a two-dimensional boundary correction with global directional information.
However, it must be lifted back to the latent dimension of the backbone before being written into the original intermediate representation.
GeoABC performs this step through an up-projection $\mathbf{W}_{\uparrow}$ and updates the intermediate representation in a residual form:
\begin{equation}
    \tilde{\mathbf{z}}_{\ell}(\mathbf{x})
    =
    \mathbf{z}_{\ell}(\mathbf{x})
    +
    \omega_G(\mathbf{x})
    \left(
        \mathbf{W}_{\uparrow} \mathbf{c}_{\mathrm{amb}}(\mathbf{x})
        +
        \mathbf{W}_a \mathbf{h}(\mathbf{x})
    \right).
    \label{eq:latent_write_back}
\end{equation}
Here, $\mathbf{W}_a\mathbf{h}(\mathbf{x})$ introduces an auxiliary correction that is consistent with the intermediate context of the backbone.
This update completes the anisotropic boundary correction at layer $\ell$.
The corrected representation $\tilde{\mathbf{z}}_{\ell}(\mathbf{x})$ is then passed to the remaining operator blocks and decoder to produce the final physical prediction.
The 3D extension is provided in Appendix~\ref{sec:appendix:adapt} and additional theoretical analysis is provided in Appendix~\ref{app:theory}.

\paragraph{Training Objective} The overall objective consists of a global prediction loss, a boundary-weighted loss, and an anchor alignment loss:
\begin{equation}
    \mathcal{L}
    =
    \mathcal{L}_{\mathrm{global}}
    +
    \lambda_b \mathcal{L}_{\mathrm{boundary}}
    +
    \lambda_a \mathcal{L}_{\mathrm{anchor}} .
    \label{eq:total_loss}
\end{equation}
Detailed formulations of the three loss terms are provided in Appendix~\ref{app:loss}.

\section{Experiments}
\label{sec:experiments}

\subsection{Experimental Setup}
\label{sec:exp_setup}

\paragraph{Datasets.}
We evaluate GeoABC on two steady-state aerodynamic benchmarks.
The 2D NACA airfoil dataset~\citep{li2023geofno} contains 1000 training samples and 200 test samples, where the task is to predict $(u,v,p)$ around airfoil geometries.
The 3D Shape-Net Car dataset~\citep{umetani2018learning} contains 789 training cars and 100 test cars, where the task is to predict $(u,v,w,p)$ around car geometries.
Additional dataset details are provided in Appendix~\ref{sec:appendix:datasets}.


\paragraph{Baselines.}
We evaluate GeoABC as a plug-in correction module on five representative neural operator backbones spanning four operator families: spectral (Geo-FNO~\citep{li2023geofno}), convolutional (CNO~\citep{raonic2023cno}), graph/kernel (GINO~\citep{li2023gino}), and transformer-based (GNOT~\citep{hao2023gnot}, Transolver~\citep{wu2024transolver}).
This coverage allows us to test whether GeoABC can be inserted into different architectures without changing their operator-learning formulation.
Per-backbone configurations and detailed descriptions of each baseline are provided in Appendix~\ref{sec:appendix:hparams}.


\paragraph{Metrics.}
We separately report global-field accuracy and near-wall accuracy.
For the 2D NACA benchmark, we use global relative $L_2$ and near-wall relative $L_2$ as primary metrics, with component-wise near-wall errors for $u$, $v$, $p$ and boundary pressure metrics ($C_p$ MAE, wall-pressure MAE) as secondary metrics.
For the 3D Shape-Net Car benchmark, we report volumetric field errors, wetted-surface metrics including surface velocity relative $L_2$ and surface pressure MAE, and the drag coefficient error $C_D$.
Detailed metric definitions are provided in Appendix~\ref{sec:appendix:metrics}.


\subsection{Main Results}
\label{sec:main_results}
\begin{table}[t]
\centering
\caption{Cross-backbone performance on the 2D NACA airfoil benchmark.}
\label{tab:cross_backbone}
\setlength{\tabcolsep}{3pt}
\resizebox{\linewidth}{!}{
\begin{tabular}{l@{\hskip 5pt}cccccc}
\toprule
Method
& Global rel-$L_2$ $\downarrow$
& Near-wall rel-$L_2$ $\downarrow$
& $u$ near rel-$L_2$ $\downarrow$
& $v$ near rel-$L_2$ $\downarrow$
& $p$ near rel-$L_2$ $\downarrow$
& $C_p$ MAE $\downarrow$ \\
\midrule
\textbf{Geo-FNO} & 
$0.0608\!\pm\!.003$ & $0.1320\!\pm\!.006$ & $0.1266\!\pm\!.006$ & $0.2803\!\pm\!.024$ & $0.1343\!\pm\!.006$ & $0.2085\!\pm\!.009$ \\
\quad $+$ GeoABC  & 
$\mathbf{0.0297\!\pm\!.001}$ & 
$\mathbf{0.0369\!\pm\!.001}$ & 
$\mathbf{0.0360\!\pm\!.001}$ & 
$\mathbf{0.0716\!\pm\!.003}$ & 
$\mathbf{0.0370\!\pm\!.002}$ & 
$\mathbf{0.0426\!\pm\!.004}$ \\
\midrule
\textbf{CNO} & $0.0194\!\pm\!.004$ & $0.0398\!\pm\!.009$ & $0.0392\!\pm\!.010$ & $0.0975\!\pm\!.013$ & $0.0379\!\pm\!.009$ & $0.0553\!\pm\!.014$ \\
\quad $+$ GeoABC& 
$\mathbf{0.0171\!\pm\!.003}$ & 
$\mathbf{0.0242\!\pm\!.004}$ & 
$\mathbf{0.0237\!\pm\!.004}$ & 
$\mathbf{0.0569\!\pm\!.003}$ & 
$\mathbf{0.0236\!\pm\!.004}$ & 
$\mathbf{0.0270\!\pm\!.003}$ \\
\midrule
\textbf{GNOT} & $0.0333\!\pm\!.009$ & $0.0715\!\pm\!.020$ & $0.0684\!\pm\!.019$ & $0.1794\!\pm\!.043$ & $0.0712\!\pm\!.020$ & $0.1013\!\pm\!.026$ \\
\quad $+$ GeoABC & 
$\mathbf{0.0240\!\pm\!.001}$ & 
$\mathbf{0.0396\!\pm\!.003}$ & 
$\mathbf{0.0388\!\pm\!.002}$ & 
$\mathbf{0.1074\!\pm\!.008}$ & 
$\mathbf{0.0376\!\pm\!.004}$ & 
$\mathbf{0.0433\!\pm\!.001}$ \\
\midrule
\textbf{Transolver} & $0.0359\!\pm\!.003$ & $0.0409\!\pm\!.004$ & $0.0181\!\pm\!.002$ & $0.0868\!\pm\!.007$ & $0.0177\!\pm\!.002$ & $0.0143\!\pm\!.001$ \\
\quad $+$ GeoABC & 
$\mathbf{0.0317\!\pm\!.002}$ & 
$\mathbf{0.0355\!\pm\!.002}$ & 
$\mathbf{0.0166\!\pm\!.001}$ & 
$\mathbf{0.0737\!\pm\!.005}$ & 
$\mathbf{0.0162\!\pm\!.001}$ & 
$\mathbf{0.0118\!\pm\!.001}$ \\
\bottomrule
\end{tabular}
}
\end{table}

\begin{figure}[t]
    \centering
    \includegraphics[width=0.95\linewidth]{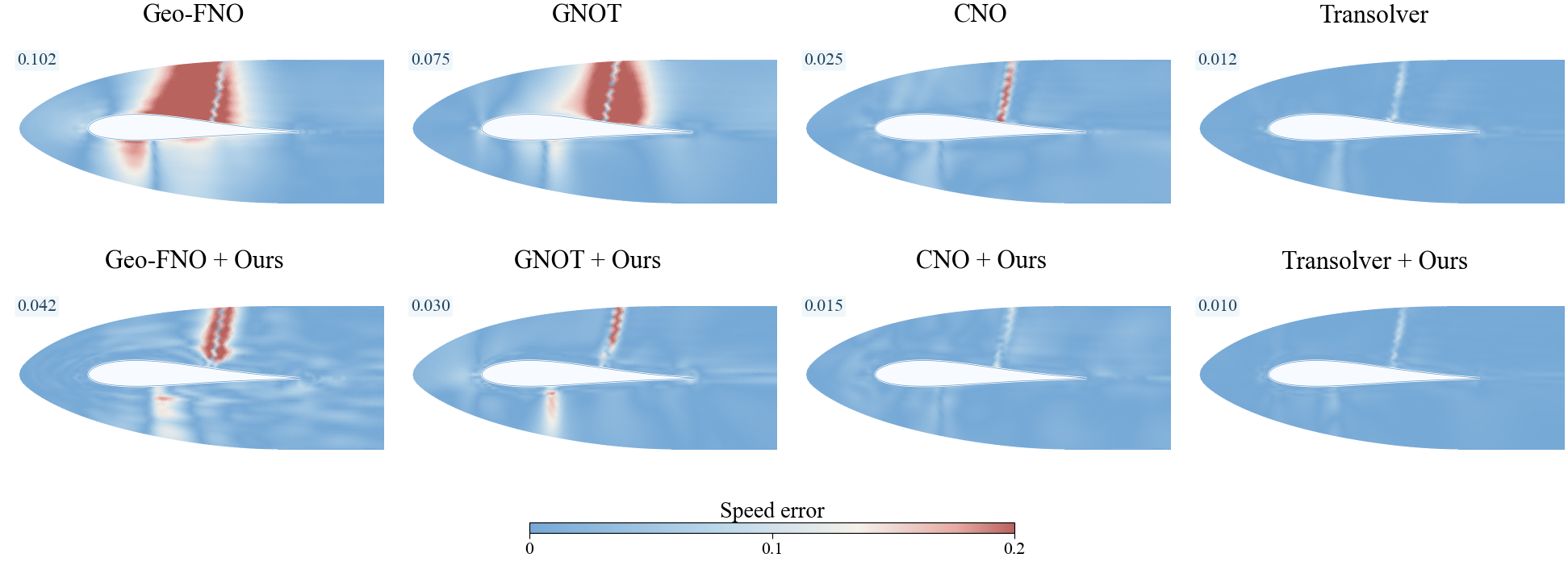}
\caption{
Velocity-magnitude prediction and error maps on the 2D NACA benchmark.
The numbers in the error maps denote near-wall relative $\ell_2$ error.
}   

    \label{fig:2d_vis_main}
\end{figure}

\paragraph{2D prediction.}

Before evaluating GeoABC, we first examine whether standard boundary-aware remedies can close the near-wall gap.
As shown in Fig.~\ref{fig:intro}, our diagnostic experiments show that these methods do not effectively reduce near-wall error and can even degrade prediction on stronger backbones.
This indicates that the near-wall gap cannot be resolved by simply enhancing boundary awareness or adding boundary-specific capacity. Detailed settings and results are provided in Appendix~\ref{sec:appendix:boundary_awareness}.

\Cref{tab:cross_backbone} shows that GeoABC delivers substantial and consistent gains across all four 2D backbones.
Across spectral, convolutional, and attention-based operators, GeoABC reduces near-wall rel-$L_2$ by $13.1\%$--$72.0\%$ and lowers $C_p$ MAE by $17.4\%$--$79.6\%$.
The improvement appears not only on Geo-FNO, but also on CNO, GNOT, and the strongest 2D backbone, Transolver, suggesting that the near-wall gap is not merely a symptom of insufficient backbone capacity.

\Cref{fig:2d_vis_main} further visualizes the speed-error fields.
The original backbones show pronounced wall-normal error bands near the airfoil surface, indicating a directional near-wall error pattern.
GeoABC consistently suppresses these bands across backbones, with especially clear reductions on Geo-FNO and GNOT.
This supports that GeoABC corrects anisotropic near-wall errors rather than merely improving the field uniformly.

The per-metric pattern is consistent with this visual evidence.
On every backbone, the relative reductions in near-wall rel-$L_2$ and $C_p$ MAE are larger than the reduction in global rel-$L_2$.
Therefore, GeoABC mainly improves boundary-sensitive quantities, especially $C_p$, whose accuracy directly depends on pressure prediction along the airfoil surface, reflecting the effective constraint imposed by $\omega_G(x)$.
Overall, the consistent gains across different operator families support GeoABC as a plug-in anisotropic boundary correction module rather than a backbone-specific modification.

\begin{table}[t]
\centering
\footnotesize
\caption{Cross-backbone performance on the 3D Car benchmark.}
\label{tab:car3d}
\setlength{\tabcolsep}{3pt}
\resizebox{\linewidth}{!}{
\begin{tabular}{l@{\hskip 5pt}cccccc}
\toprule
Method
& Volume $v$ rel-$L_2$ $\downarrow$
& Volume $p$ rel-$L_2$ $\downarrow$
& Surface $v$ rel-$L_2$ $\downarrow$
& Surface $p$ MAE $\downarrow$
& $C_D$ error $\downarrow$
& $\rho_D$ $\uparrow$ \\
\midrule
\textbf{Geo-FNO} & $0.0365\!\pm\!.000$ & $0.1103\!\pm\!.000$ & $0.0447\!\pm\!.000$ & $3.177\!\pm\!.008$ & $0.0223\!\pm\!.001$ & $0.9649\!\pm\!.001$ \\
\quad $+$ GeoABC & 
$\mathbf{0.0355\!\pm\!.000}$ & 
$\mathbf{0.1030\!\pm\!.000}$ & 
$\mathbf{0.0397\!\pm\!.000}$ & 
$\mathbf{2.922\!\pm\!.007}$ & 
$0.0233\!\pm\!.001$ & 
$\mathbf{0.9712\!\pm\!.001}$ \\
\midrule
\textbf{GINO} & $0.0355\!\pm\!.001$ & $0.1153\!\pm\!.004$ & $0.0435\!\pm\!.001$ & $3.288\!\pm\!.080$ & $0.0862\!\pm\!.002$ & $0.8093\!\pm\!.006$ \\
\quad $+$ GeoABC & 
$\mathbf{0.0326\!\pm\!.000}$ & 
$\mathbf{0.1030\!\pm\!.000}$ & 
$\mathbf{0.0360\!\pm\!.000}$ & 
$\mathbf{2.912\!\pm\!.001}$ & 
$\mathbf{0.0794\!\pm\!.003}$ & 
$\mathbf{0.8249\!\pm\!.005}$ \\
\midrule
\textbf{GNOT} & $0.0424\!\pm\!.000$ & $0.0887\!\pm\!.001$ & $0.0655\!\pm\!.000$ & $2.314\!\pm\!.024$ & $0.0149\!\pm\!.001$ & $0.9888\!\pm\!.001$ \\
\quad $+$ GeoABC & 
$\mathbf{0.0276\!\pm\!.000}$ & 
$\mathbf{0.0862\!\pm\!.002}$ & 
$\mathbf{0.0322\!\pm\!.001}$ & 
$\mathbf{2.230\!\pm\!.034}$ & 
$\mathbf{0.0138\!\pm\!.001}$ & 
$\mathbf{0.9890\!\pm\!.002}$ \\
\midrule
\textbf{Transolver} & $0.0433\!\pm\!.001$ & $0.0854\!\pm\!.000$ & $0.0656\!\pm\!.001$ & $2.299\!\pm\!.010$ & $0.0152\!\pm\!.003$ & $0.9859\!\pm\!.001$ \\
\quad $+$ GeoABC & 
$\mathbf{0.0273\!\pm\!.001}$ & 
$\mathbf{0.0823\!\pm\!.000}$ & 
$\mathbf{0.0314\!\pm\!.000}$ & 
$\mathbf{2.078\!\pm\!.012}$ & 
$\mathbf{0.0143\!\pm\!.001}$ & 
$\mathbf{0.9892\!\pm\!.002}$ \\
\bottomrule
\end{tabular}
}
\end{table}

\paragraph{3D prediction.}
We further evaluate whether GeoABC can extend to complex 3D car geometries.
As shown in \Cref{tab:car3d}, GeoABC consistently improves boundary-related metrics across all 3D backbones.
Surface velocity rel-$L_2$ decreases on every backbone, with reductions ranging from $11.2\%$ to $52.1\%$, and surface pressure MAE is also reduced across all backbones.
These improvements are concentrated on wetted-surface quantities, which are the most relevant quantities for boundary-accurate aerodynamic prediction.

\Cref{fig:3d_vis} further visualizes the 3D velocity field.
Compared with the original backbones, GeoABC produces surrounding streamlines that are closer to the ground truth and reduces the visualization error across multiple backbones. Additional visualizations are provided in Appendix~\ref{sec:appendix:visualization}.

The improvement further transfers to downstream aerodynamic metrics.
GeoABC reduces the drag-coefficient error $C_D$ on GINO, GNOT, and Transolver, while preserving or improving the Spearman correlation $\rho_D$ between predicted and ground-truth $C_D$ across all backbones.
Since $C_D$ and $\rho_D$ measure whether the surrogate can support drag-based aerodynamic evaluation and design ranking, these results indicate that improving boundary-region prediction leads to more reliable engineering-level predictions.
Overall, the 3D results show that GeoABC does not rely on the 2D airfoil geometry, but remains effective as a plug-in boundary correction module for complex 3D aerodynamic geometries.

\begin{figure}[t]
    \centering
    \includegraphics[width=\linewidth]{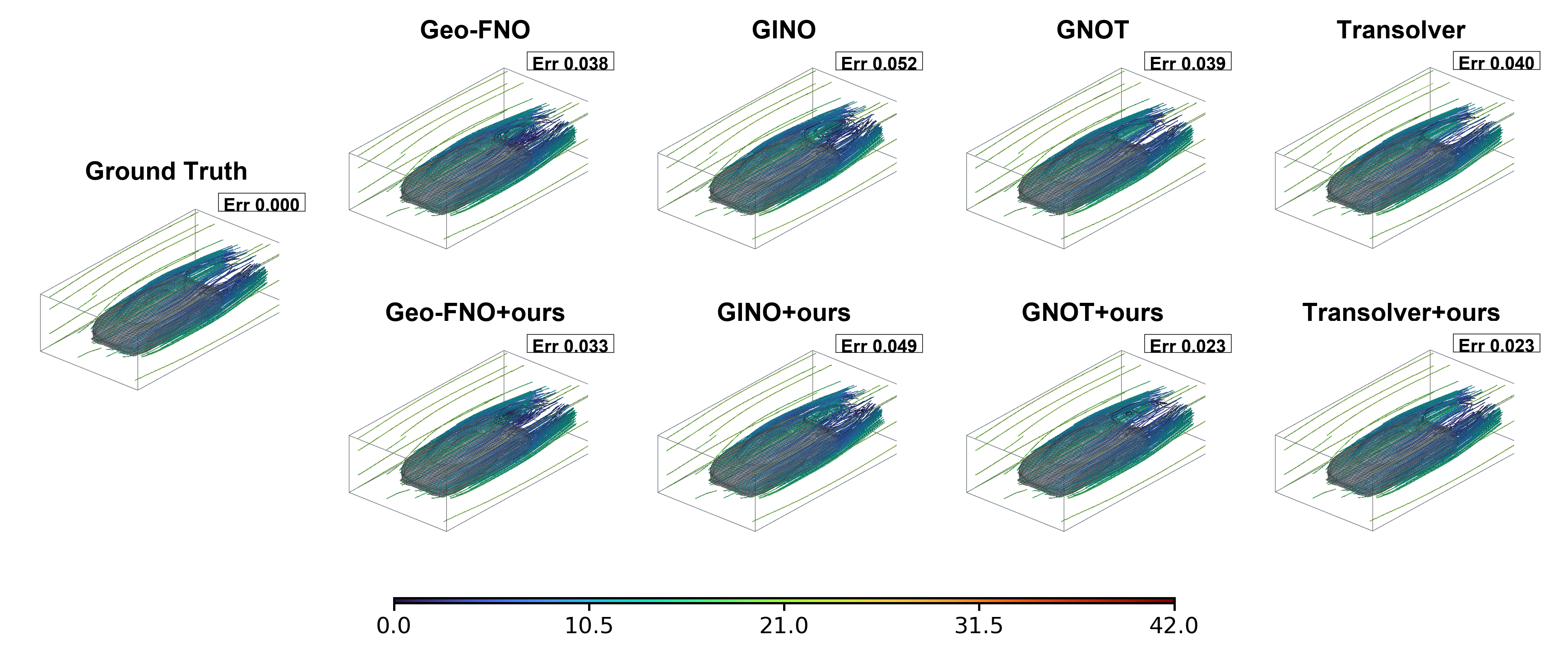}
    \caption{
3D visualization of the surrounding velocity field around the car, showing airflow streamlines and the vehicle surface.
}
    \label{fig:3d_vis}
\end{figure}

\subsection{Ablation Study}
\label{sec:ablation}
\label{sec:exp:isolation}

Our ablation study, summarized in Fig.~\ref{fig:ablation_radar}, evaluates the contribution of each design component in GeoABC.
We compare the full model with four ablated variants.
\textbf{no-inject} removes the mid-layer boundary injection module while keeping the same training losses, testing whether the gain comes merely from boundary-related supervision.
\textbf{no-frame} keeps the boundary-correction module but removes the local tangent--normal frame, testing whether the gain comes from the structural prior of solid boundaries rather than boundary-local correction capacity alone.
\textbf{normal-only} keeps the local frame but uses only the wall-normal branch, removing the tangential and cross-direction branches.
\textbf{zero-geometry} keeps the boundary-correction structure but removes the geometry-conditioned spatial gate, testing whether the model needs geometry to localize where the correction should be activated.


\begin{figure}[t]
    \centering
    \includegraphics[width=1.0\linewidth]{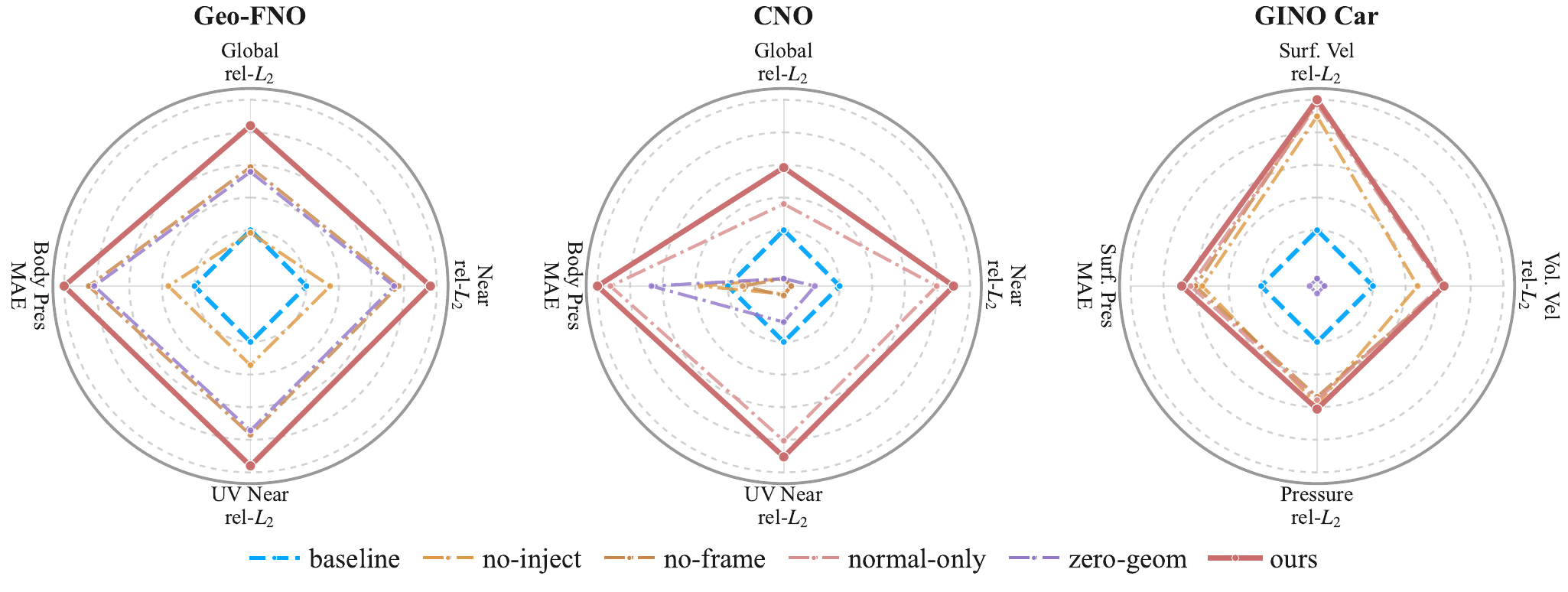}
    \caption{Component isolation of GeoABC, comparing the full GeoABC with four ablated variants on representative 2D airfoil and 3D car experiments.}
    \label{fig:ablation_radar}
\end{figure}


As shown in Fig.~\ref{fig:ablation_radar}, the full GeoABC consistently outperforms the ablated variants.
The weak performance of \textbf{no-inject} shows that boundary-related losses alone are not sufficient to explain the improvement.
The degradation of \textbf{no-frame} shows that boundary correction must be organized by the local directions induced by solid boundaries, rather than applied in an unstructured latent space.
\textbf{normal-only} recovers part of the gain, confirming the importance of wall-normal responses near solid boundaries; however, the full model remains stronger overall, suggesting that tangential and cross-directional corrections also contribute to modeling anisotropic near-wall physics.


The \textbf{zero-geometry} variant further shows that geometry-conditioned spatial gating is needed to localize where the correction should be activated.
Together, these ablations show that GeoABC’s gains come not from extra boundary-local capacity alone, but from the joint use of mid-layer correction, solid-boundary directional structure, anisotropic branches, and geometry-conditioned spatial gating. Detailed numerical results and additional ablation analyses are provided in Appendix~\ref{sec:appendix:component_details}.


\subsection{Additional Experiments}
\label{sec:additional_experimental_evidence}

Beyond the main results and ablations, we provide additional experiments in the appendix.
Appendix~\ref{sec:appendix:spatial_attribution} shows that the improvement of GeoABC is a general boundary-band effect rather than an accidental gain from a specific local region; it also shows that the near-wall improvement partly propagates to the surrounding flow field, further supporting mid-layer correction.
Appendix~\ref{sec:appendix:pipe} further evaluates GeoABC on a new Pipe NS task, showing that the method does not rely on curved geometries and can also transfer to straight no-slip walls.
Appendix~\ref{sec:appendix:injection_position} shows that middle-layer insertion is the most effective choice, where the backbone has already formed sufficient physical context while the corrected representation can still interact with other regions through the remaining operator layers.
Appendix~\ref{sec:appendix:lambda_b_band_sensitivity} shows that the improvement of GeoABC is stable across boundary-loss weights and boundary-band definitions, indicating that the near-wall gain is not caused by a finely tuned loss weight or a particular narrow-band choice.

\section{Conclusion}
\label{sec:conclusion}

This paper presents GeoABC, which explicitly leverages solid boundary geometry as a structural prior to mitigate the prediction errors caused by isotropic near-wall modeling in aerodynamic neural operators. GeoABC consists of two main mechanisms: local tangent-normal decomposition, which characterizes the difference between tangential propagation along the wall and strong wall-normal variations, and geometry-gated boundary correction, which performs direction-aware modulation within trusted boundary regions in the intermediate representation space. Experiments show that GeoABC can be integrated into various neural operator backbones as a plug-in module, consistently reducing near-wall errors and improving boundary-related metrics on both 2D airfoil and 3D car tasks. This work demonstrates that injecting the anisotropic structure of classical boundary physics into neural operators can improve their reliability in high-fidelity aerodynamic simulation and engineering-level aerodynamic metric prediction.


\bibliography{references}
\bibliographystyle{unsrt}

\newpage
\appendix
\section{Theoretical Properties of GeoABC}
\label{app:theory}

This appendix formalizes GeoABC as a geometry-gated near-wall boundary
correction and records three structural properties of the resulting
representation. The goal is not to claim that a particular injection
layer, branch design, or training strategy is universally optimal, nor
to provide a stability guarantee for the underlying operator-learning
problem. Instead, we record (i) a write-back locality property by
construction, (ii) a Lipschitz-type spatial regularity bound on the
corrected representation under standard regularity assumptions, and
(iii) an identity-reachability property which ensures that GeoABC does
not reduce the approximation capacity of the original backbone. We
view these properties as sanity checks that clarify what GeoABC does
and does not change at the representation level; the empirical claims
of the paper are supported by the experiments in
Sec.~\ref{sec:experiments}, not by these properties.

\subsection{Anisotropy Motivation}

Near a solid boundary $\partial\Omega_G$, aerodynamic fields typically
vary differently along the tangential and normal directions, and the
appropriate notion of anisotropy depends on the flow regime.

\paragraph{Viscous regime.} For viscous flows, classical Prandtl
boundary-layer scaling gives
\begin{equation}
    \delta \sim L\,\mathrm{Re}^{-1/2},
\end{equation}
where $L$ is the characteristic length, $\mathrm{Re}$ is the Reynolds
number, and $\delta$ is the boundary-layer thickness. If the
tangential velocity has the boundary-layer form
$u_t(s,n)=U(s/L,\,n/\delta)$, the chain rule gives
\begin{equation}
    \left|\partial_n u_t\right| \sim 1/\delta, \qquad
    \left|\partial_s u_t\right| \sim 1/L,
\end{equation}
so the scale ratio between normal and tangential variations is
\begin{equation}
    \frac{|\partial_n u_t|}{|\partial_s u_t|}
    \sim L/\delta \sim \mathrm{Re}^{1/2} \gg 1.
\end{equation}
This applies to the 3D RANS car benchmark used in our experiments.

\paragraph{Inviscid regime.} For the 2D transonic Euler benchmark, the
no-penetration slip wall does not produce a Prandtl boundary layer,
and the tangential wall velocity is generally nonzero. The relevant
near-wall anisotropy is instead driven by wall-tangential pressure
gradients, suction-peak structures, and possible shock-induced normal
jumps. In both cases, near-wall fields exhibit different magnitudes
of variation along the wall-tangent and wall-normal directions, which
motivates organizing the boundary correction in a local tangent-normal
frame
\begin{equation}
    \mathbf{B}_G(\mathbf{x})
    = [\hat{\mathbf{t}}(\mathbf{x}),\,\hat{\mathbf{n}}(\mathbf{x})]
\end{equation}
rather than in the ambient coordinate frame. We do not claim a
quantitative scale law that applies uniformly across both regimes;
the role of the local frame is structural, not asymptotic.

\subsection{Notation}

We follow the notation in Sec.~\ref{sec:method:correction}. For a
geometry $G$, let $\Omega_G$ denote the fluid domain and
$\partial\Omega_G$ its solid boundary. The near-wall band is
\begin{equation}
    \Omega_{\delta}(G)
    = \{\mathbf{x}\in\Omega_G:
       \operatorname{dist}(\mathbf{x},\partial\Omega_G)<\delta\}.
\end{equation}
GeoABC precomputes the local geometric descriptor
\begin{equation}
    \mathbf{g}_G(\mathbf{x})
    = \big(\phi(\mathbf{x}),\,
           \hat{\mathbf{t}}(\mathbf{x}),\,\hat{\mathbf{n}}(\mathbf{x}),\,
           \kappa(\mathbf{x}),\,
           w(\mathbf{x}),\,c(\mathbf{x})\big),
\end{equation}
where $\phi$ is the signed distance to $\partial\Omega_G$,
$\hat{\mathbf{t}},\hat{\mathbf{n}}$ are local unit tangent and normal
vectors satisfying $\hat{\mathbf{t}}\perp\hat{\mathbf{n}}$ and
$\|\hat{\mathbf{t}}\|=\|\hat{\mathbf{n}}\|=1$ wherever they are
defined, $\kappa$ is the local curvature,
$w(\mathbf{x})\in[0,1]$ is the near-wall window, and
$c(\mathbf{x})\in[0,1]$ is the confidence score of the geometric
descriptor. The geometry-confidence gate is
\begin{equation}
    \omega_G(\mathbf{x}) = w(\mathbf{x})\,c(\mathbf{x}).
\end{equation}
The scalar $c(\mathbf{x})$ should not be confused with
$\mathbf{c}_B(\mathbf{x})$, which denotes the local-frame correction
introduced below.

\paragraph{Geometric singularities.}
The frame $\hat{\mathbf{t}},\hat{\mathbf{n}}$ is defined via the
nearest-point projection onto $\partial\Omega_G$, which is
single-valued and smooth away from the medial axis of $\Omega_G$ and
away from non-smooth points of $\partial\Omega_G$ (e.g., the trailing
edge of an airfoil). Throughout this appendix, regularity statements
are made on compact subsets $K\subset\Omega_G$ that exclude a
neighborhood of these singular sets. In practice, the confidence
score $c(\mathbf{x})$ is designed to vanish smoothly near such
regions, which is consistent with this restriction.

Let $\mathbf{z}_{\ell}(\mathbf{x})\in\mathbb{R}^{d_z}$ be the
backbone's intermediate representation at layer $\ell$. GeoABC
projects this representation into a low-dimensional anchor space,
\begin{equation}
    \mathbf{a}(\mathbf{x}) = \mathbf{W}_{\downarrow}\,
    \mathbf{z}_{\ell}(\mathbf{x}) \in \mathbb{R}^{2},
\end{equation}
and extracts an auxiliary context vector
\begin{equation}
    \mathbf{h}(\mathbf{x}) = \mathbf{W}_{h}\,\mathbf{z}_{\ell}(\mathbf{x}).
\end{equation}
In the 2D setting the local tangent-normal frame is
$\mathbf{B}_G(\mathbf{x})=[\hat{\mathbf{t}}(\mathbf{x}),
\hat{\mathbf{n}}(\mathbf{x})]\in\mathbb{R}^{2\times 2}$, and the anchor
representation is projected into this local frame as
\begin{equation}
    \mathbf{a}_B(\mathbf{x})
    = \mathbf{B}_G(\mathbf{x})^{\top}\mathbf{a}(\mathbf{x})
    = \begin{bmatrix} a_t(\mathbf{x}) \\ a_n(\mathbf{x}) \end{bmatrix}
      \in \mathbb{R}^{2}.
\end{equation}
Because $\hat{\mathbf{t}}\perp\hat{\mathbf{n}}$ and both are unit
length, $\mathbf{B}_G$ is orthogonal, so
$\mathbf{B}_G^{\top}=\mathbf{B}_G^{-1}$. GeoABC computes three scalar
branch outputs
\begin{align}
    r_t(\mathbf{x}) &= f_t\!\big(a_t(\mathbf{x}),\,\mathbf{h}(\mathbf{x})\big), \\
    r_n(\mathbf{x}) &= f_n\!\big(a_n(\mathbf{x}),\,\mathbf{h}(\mathbf{x})\big), \\
    r_m(\mathbf{x}) &= f_m\!\big(a_t(\mathbf{x}),a_n(\mathbf{x}),\,\mathbf{h}(\mathbf{x})\big),
\end{align}
where the subscript $m$ denotes the cross tangent-normal branch.
Geometry-conditioned gates are
\begin{equation}
    \big(\alpha_t(\mathbf{x}),\,\alpha_n(\mathbf{x}),\,\alpha_m(\mathbf{x})\big)
    = \sigma\!\left(\Gamma_{\psi}\big(\mathbf{g}_G(\mathbf{x})\big)\right).
\end{equation}
The cross branch is lifted to a 2D correction in the local frame,
$\mathbf{q}_m(\mathbf{x})=\mathbf{W}_m r_m(\mathbf{x})$, and the
local-frame correction is
\begin{equation}
    \mathbf{c}_B(\mathbf{x})
    = \begin{bmatrix}
        \alpha_t(\mathbf{x})\,r_t(\mathbf{x}) \\
        \alpha_n(\mathbf{x})\,r_n(\mathbf{x})
      \end{bmatrix}
      + \alpha_m(\mathbf{x})\,\mathbf{q}_m(\mathbf{x}).
\end{equation}
Mapping back to the ambient anchor space gives
$\mathbf{c}_{\mathrm{amb}}(\mathbf{x})
=\mathbf{B}_G(\mathbf{x})\,\mathbf{c}_B(\mathbf{x})$, and the latent
residual produced by GeoABC is
\begin{equation}
    \Delta\mathbf{z}_{\ell}(\mathbf{x})
    = \mathbf{W}_{\uparrow}\,\mathbf{c}_{\mathrm{amb}}(\mathbf{x})
    + \mathbf{W}_a\,\mathbf{h}(\mathbf{x}).
\end{equation}
Finally, GeoABC writes the residual back to the intermediate
representation as
\begin{equation}
    \tilde{\mathbf{z}}_{\ell}(\mathbf{x})
    = \mathbf{z}_{\ell}(\mathbf{x})
    + \omega_G(\mathbf{x})\,\Delta\mathbf{z}_{\ell}(\mathbf{x}).
    \label{eq:writeback}
\end{equation}
The directional residual $\Delta\mathbf{z}_{\ell}$ determines what is
corrected, and the geometry-confidence gate $\omega_G$ determines
where the correction is applied and how strongly.

\subsection{Boundary-Band Local Write-Back}

The first property is a direct consequence of the write-back rule
\eqref{eq:writeback} and is included only to make explicit how
$\omega_G$ controls the spatial support of the correction at the
injection layer. The statement is given in two forms, depending on
whether the window $w(\mathbf{x})$ is implemented as a compactly
supported (hard) window or as a smooth (soft) window that decays away
from the boundary.

\begin{proposition}[Write-back locality at the injection layer]
\label{prop:locality}
The following hold pointwise for the write-back rule
\eqref{eq:writeback}.
\begin{enumerate}
\item[(i)] \emph{Pointwise identity.}
If $\omega_G(\mathbf{x})=0$, then
$\tilde{\mathbf{z}}_{\ell}(\mathbf{x})=\mathbf{z}_{\ell}(\mathbf{x})$.
\item[(ii)] \emph{Pointwise magnitude bound.}
Since $w(\mathbf{x}),c(\mathbf{x})\in[0,1]$,
\begin{equation}
    \|\tilde{\mathbf{z}}_{\ell}(\mathbf{x})
      - \mathbf{z}_{\ell}(\mathbf{x})\|
    \le \|\Delta\mathbf{z}_{\ell}(\mathbf{x})\|
    \quad\text{for all } \mathbf{x}\in\Omega_G.
\end{equation}
\item[(iii)] \emph{Hard window.}
If the window has compact support inside the near-wall band, i.e.\
$\operatorname{supp}(\omega_G)\subseteq\Omega_{\delta}(G)$, then
$\tilde{\mathbf{z}}_{\ell}(\mathbf{x})=\mathbf{z}_{\ell}(\mathbf{x})$
for every $\mathbf{x}\in\Omega_G\setminus\Omega_{\delta}(G)$.
\item[(iv)] \emph{Soft window.}
More generally, suppose $\omega_G$ is bounded outside the near-wall
band by $\omega_G(\mathbf{x})\le\varepsilon$ for all
$\mathbf{x}\in\Omega_G\setminus\Omega_{\delta}(G)$, and that
$\|\Delta\mathbf{z}_{\ell}(\mathbf{x})\|\le D_{\Delta}$ on a set
$K\subset\Omega_G$. Then for every
$\mathbf{x}\in K\cap\big(\Omega_G\setminus\Omega_{\delta}(G)\big)$,
\begin{equation}
    \|\tilde{\mathbf{z}}_{\ell}(\mathbf{x})
      - \mathbf{z}_{\ell}(\mathbf{x})\|
    \le \varepsilon\,D_{\Delta}.
\end{equation}
\end{enumerate}
\end{proposition}

\begin{proof}
By \eqref{eq:writeback},
$\tilde{\mathbf{z}}_{\ell}(\mathbf{x})-\mathbf{z}_{\ell}(\mathbf{x})
=\omega_G(\mathbf{x})\,\Delta\mathbf{z}_{\ell}(\mathbf{x})$.
Statement~(i) is immediate. For~(ii), $w,c\in[0,1]$ gives
$0\le\omega_G\le 1$, and absolute homogeneity of the norm yields
$\|\omega_G(\mathbf{x})\Delta\mathbf{z}_{\ell}(\mathbf{x})\|
=\omega_G(\mathbf{x})\|\Delta\mathbf{z}_{\ell}(\mathbf{x})\|
\le\|\Delta\mathbf{z}_{\ell}(\mathbf{x})\|$. Statement~(iii) is the
contrapositive of~(i) under the support assumption. For~(iv),
$\omega_G(\mathbf{x})\le\varepsilon$ and the same homogeneity argument
give the bound $\varepsilon D_{\Delta}$.
\end{proof}

We emphasize that Proposition~\ref{prop:locality} is an
identity-by-construction property of the write-back rule, not a
non-trivial localization result for the final prediction; it states
that GeoABC does not directly modify the intermediate representation
at points where $\omega_G$ vanishes (or is small), and that for soft
windows the residual write-back outside the near-wall band is
controlled by the gate's far-band magnitude $\varepsilon$. The
behavior of the final output at bulk points is discussed separately
in Sec.~\ref{sec:bulk-remark}.

\subsection{Spatial Regularity of the Corrected Representation}

The next property records that, under standard regularity assumptions
on a compact set $K$ avoiding the geometric singularities discussed
above, GeoABC's write-back preserves the Lipschitz behavior of the
intermediate representation, with a quantitative bound. We state the
result as a regularity bound on $\tilde{\mathbf{z}}_{\ell}$ as a
function of the spatial query $\mathbf{x}$, and not as a stability
guarantee for the operator $\tilde{\mathcal{F}}_{\theta,\psi}$ with
respect to perturbations of its functional inputs.

\paragraph{Setting.}
All regularity statements in this subsection are made for a fixed
geometry $G$, a fixed operating condition $s$, a fixed
discretization or query set, and fixed model parameters
$(\theta,\psi)$. We suppress these dependencies and view the hidden
representation $\mathbf{z}_{\ell}$, the residual
$\Delta\mathbf{z}_{\ell}$, and the gate $\omega_G$ as functions of
the spatial query $\mathbf{x}\in\Omega_G$ only. The Lipschitz
constants below should therefore be read as spatial-regularity
constants of these functions, not as stability constants of the
underlying operator with respect to its inputs.

\begin{theorem}[Lipschitz bound on the corrected representation]
\label{thm:lipschitz}
Let $K\subset\Omega_G$ be compact and disjoint from a neighborhood of
the medial axis of $\Omega_G$ and from non-smooth points of
$\partial\Omega_G$. Assume there exist constants
$L_{\mathbf{z}_{\ell}},L_{\omega},L_{\Delta\mathbf{z}},D_{\Delta}\ge 0$
such that
\begin{align}
    \|\mathbf{z}_{\ell}(\mathbf{x})-\mathbf{z}_{\ell}(\mathbf{y})\|
    &\le L_{\mathbf{z}_{\ell}}\|\mathbf{x}-\mathbf{y}\|, \\
    |\omega_G(\mathbf{x})-\omega_G(\mathbf{y})|
    &\le L_{\omega}\|\mathbf{x}-\mathbf{y}\|, \\
    \|\Delta\mathbf{z}_{\ell}(\mathbf{x})
      -\Delta\mathbf{z}_{\ell}(\mathbf{y})\|
    &\le L_{\Delta\mathbf{z}}\|\mathbf{x}-\mathbf{y}\|, \\
    \|\Delta\mathbf{z}_{\ell}(\mathbf{x})\| &\le D_{\Delta},
\end{align}
for all $\mathbf{x},\mathbf{y}\in K$. Then
$\tilde{\mathbf{z}}_{\ell}$ is Lipschitz on $K$ with
\begin{equation}
    \mathrm{Lip}_K(\tilde{\mathbf{z}}_{\ell})
    \le L_{\mathbf{z}_{\ell}}
        + L_{\omega}D_{\Delta}
        + L_{\Delta\mathbf{z}}.
    \label{eq:lip-z}
\end{equation}
If, in addition, the post-injection map (the remaining backbone layers
and decoder after layer $\ell$) restricted to the relevant range is
$L_{\mathrm{post}}$-Lipschitz, then on $K$ the predicted field
$\tilde{\mathcal{F}}_{\theta,\psi}=P_{\theta}\circ
\tilde{\mathbf{z}}_{\ell}$ satisfies
\begin{equation}
    \mathrm{Lip}_K(\tilde{\mathcal{F}}_{\theta,\psi})
    \le L_{\mathrm{post}}\!\left(
        L_{\mathbf{z}_{\ell}}
        + L_{\omega}D_{\Delta}
        + L_{\Delta\mathbf{z}}
    \right).
    \label{eq:lip-F}
\end{equation}
\end{theorem}

\begin{proof}
Define $\delta\mathbf{z}_{\ell}(\mathbf{x})
=\omega_G(\mathbf{x})\,\Delta\mathbf{z}_{\ell}(\mathbf{x})$. For
$\mathbf{x},\mathbf{y}\in K$, adding and subtracting
$\omega_G(\mathbf{y})\,\Delta\mathbf{z}_{\ell}(\mathbf{x})$ and
applying the triangle inequality gives
\begin{align}
    \|\delta\mathbf{z}_{\ell}(\mathbf{x})
      -\delta\mathbf{z}_{\ell}(\mathbf{y})\|
    &\le |\omega_G(\mathbf{x})-\omega_G(\mathbf{y})|\,
         \|\Delta\mathbf{z}_{\ell}(\mathbf{x})\| \nonumber\\
    &\quad + |\omega_G(\mathbf{y})|\,
         \|\Delta\mathbf{z}_{\ell}(\mathbf{x})
           -\Delta\mathbf{z}_{\ell}(\mathbf{y})\|.
\end{align}
Using $|\omega_G(\mathbf{y})|\le 1$ and the four assumed bounds yields
$\|\delta\mathbf{z}_{\ell}(\mathbf{x})
 -\delta\mathbf{z}_{\ell}(\mathbf{y})\|
\le (L_{\omega}D_{\Delta}+L_{\Delta\mathbf{z}})\|\mathbf{x}-\mathbf{y}\|$.
Combining with the Lipschitz bound for $\mathbf{z}_{\ell}$ via
$\tilde{\mathbf{z}}_{\ell}=\mathbf{z}_{\ell}+\delta\mathbf{z}_{\ell}$
gives \eqref{eq:lip-z}. The bound \eqref{eq:lip-F} follows from the
composition rule for Lipschitz maps.
\end{proof}

\paragraph{Scope and caveats.}
Theorem~\ref{thm:lipschitz} is a regularity bound on the
\emph{predicted field as a function of the spatial coordinate}
$\mathbf{x}$, and should not be read as a stability statement for the
operator with respect to its functional inputs. Two limitations are
worth making explicit. First, the assumptions on
$\hat{\mathbf{t}},\hat{\mathbf{n}}$ (and hence on $\Delta\mathbf{z}_{\ell}$
through $\mathbf{B}_G$) generally fail near the medial axis and near
non-smooth points of $\partial\Omega_G$; the compact set $K$ is chosen
to exclude a neighborhood of these singular sets, which is consistent
with the smooth roll-off of $c(\mathbf{x})$ used in practice. Second,
$L_{\mathrm{post}}$ collects the Lipschitz behavior of the remaining
operator layers, including spectral and attention layers; tight
quantitative control of $L_{\mathrm{post}}$ is well known to be
difficult for transformer-based backbones and is not pursued here.
Theorem~\ref{thm:lipschitz} therefore plays the role of a structural
sanity check rather than a quantitative robustness certificate.

\subsection{Identity Reachability}

The third property records that GeoABC contains the original backbone
as a special case at the parameter level. This implies that the
function class realized by the GeoABC-augmented model is at least as
expressive as that of the original backbone.

\begin{proposition}[Identity reachability]
\label{prop:identity}
Let $\mathcal{F}$ denote the function class realized by the original
backbone neural operator and let $\widetilde{\mathcal{F}}$ denote the
function class realized after inserting GeoABC. Then
$\mathcal{F}\subseteq\widetilde{\mathcal{F}}$.
Consequently, any approximation property of $\mathcal{F}$ on a
function class of interest (e.g., density in a target space, or a
universal approximation property in the operator-learning sense of the
backbone's analysis) is inherited by $\widetilde{\mathcal{F}}$.
\end{proposition}

\begin{proof}
Choose $\mathbf{W}_{\uparrow}=\mathbf{0}$ and
$\mathbf{W}_a=\mathbf{0}$. Then
$\Delta\mathbf{z}_{\ell}\equiv\mathbf{0}$ pointwise, and
\eqref{eq:writeback} gives
$\tilde{\mathbf{z}}_{\ell}=\mathbf{z}_{\ell}$. The remaining
backbone layers and decoder are untouched, so the GeoABC-augmented
model coincides with the original backbone at this parameter setting.
Hence $\mathcal{F}\subseteq\widetilde{\mathcal{F}}$, and any
approximation property of $\mathcal{F}$ transfers to
$\widetilde{\mathcal{F}}$.
\end{proof}

Proposition~\ref{prop:identity} is an identity-reachability statement
by construction; it records that GeoABC does not reduce the expressive
power of the backbone, but it does not by itself imply that GeoABC
strictly extends it. The empirical improvements reported in
Sec.~\ref{sec:experiments} indicate that, for the backbones and
benchmarks considered, the additional capacity introduced by GeoABC
is in fact used by training to reduce near-wall error.

\subsection{Remark on Bulk Regions}
\label{sec:bulk-remark}

Proposition~\ref{prop:locality} states that GeoABC does not directly
modify the intermediate representation at points where $\omega_G$
vanishes (or is small) at the injection layer. It does \emph{not}
state that the final prediction is unchanged at bulk points. Many
neural operators use nonlocal post-injection layers, such as Fourier,
spectral, or attention layers, which can propagate near-wall
corrections to other locations. The precise statement is therefore:
GeoABC preserves local identity at points where $\omega_G$ vanishes,
and near-identity (with write-back magnitude bounded by
$\varepsilon\,D_{\Delta}$) at points where $\omega_G\le\varepsilon$,
\emph{at the injection layer}; the perturbation of the final
prediction at bulk points is controlled by the Lipschitz bound in
Theorem~\ref{thm:lipschitz} and by the nonlocal structure of the
post-injection map. This is consistent with the spatial-attribution
analysis in Appendix~\ref{sec:appendix:spatial_attribution}, where part of the
gain is concentrated in the near-wall band and part propagates into
the surrounding region through subsequent operator layers.

\section{Dataset Details}
\label{sec:appendix:datasets}
In this section, we provide additional details on the benchmark datasets used in our experiments, including their sources, prediction targets, train/test splits, and preprocessing procedures. As illustrated in \cref{fig:dataset_tasks}, the benchmarks cover two representative aerodynamic design scenarios. 
The 3D ShapeNetCar task models external flow around vehicle geometries, where the predicted pressure and velocity fields are used to evaluate drag-related aerodynamic performance. 
The 2D airfoil task models flow around airfoil geometries, where accurate near-wall prediction is important for estimating surface pressure and lift-related quantities.

\begin{figure}[b]
	\centering
\includegraphics[width=1.0\linewidth]{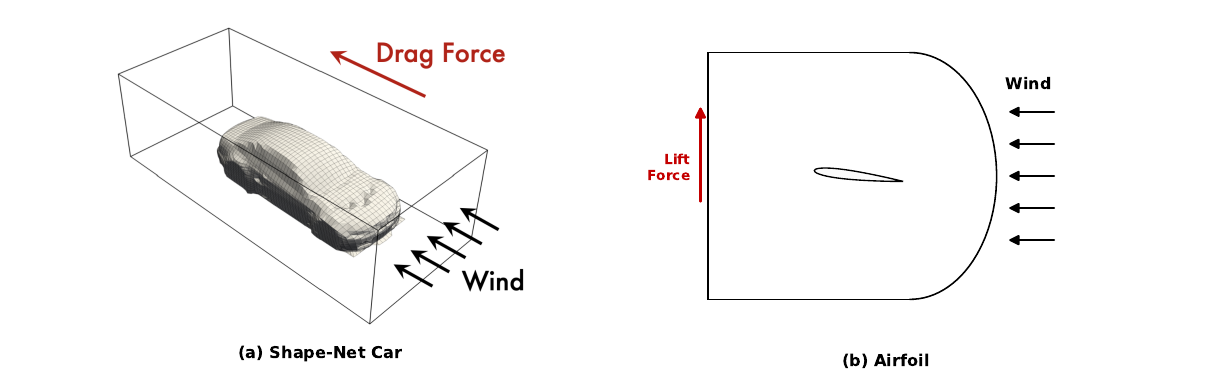}
	\caption{Car and airfoil design tasks.}
	\label{fig:dataset_tasks}
\end{figure}

\subsection{2D airfoil benchmark.}
\paragraph{Governing PDE.}
The 2D benchmark considers steady-state inviscid compressible flow around airfoil cross-sections in the transonic regime. The underlying numerical solver computes steady solutions of the compressible Euler equations,
\begin{equation*}
    \nabla \!\cdot\! \bigl(\rho \mathbf{v}\bigr) = 0,
    \qquad
    \nabla \!\cdot\! \bigl(\rho \mathbf{v}\otimes \mathbf{v} + p\mathbf{I}\bigr) = \mathbf{0},
    \qquad
    \nabla \!\cdot\! \bigl((\mathcal{E}+p)\mathbf{v}\bigr) = 0,
\end{equation*}
where $\rho$ denotes the density, $\mathbf{v}=(u,v)$ is the two-dimensional velocity vector, $p$ is the pressure, and $\mathcal{E}$ is the total energy density. For an ideal gas, the system is closed by
\begin{equation*}
    \mathcal{E}
    =
    \frac{p}{\gamma-1}
    +
    \frac{1}{2}\rho \|\mathbf{v}\|_2^2 ,
\end{equation*}
where $\gamma$ is the ratio of specific heats. Equivalently, if $E=e+\frac{1}{2}\|\mathbf{v}\|_2^2$ denotes the specific total energy, the energy equation can be written as
$\nabla\!\cdot\!\bigl((\rho E+p)\mathbf{v}\bigr)=0$.
The far-field boundary prescribes the freestream Mach number and angle of attack, while the airfoil surface satisfies an inviscid no-penetration slip-wall condition.

\paragraph{Source and generation.}
The dataset follows the airfoil benchmark released by \citep{li2023geofno}. Each sample corresponds to a deformed NACA-0012 airfoil geometry. The baseline airfoil is embedded into a design-parameter space and deformed by perturbing control nodes on the upper and lower surfaces, yielding a family of airfoil cross-sections with different geometric profiles. For each geometry, a body-fitted structured Euler solver is used to compute the corresponding steady transonic flow field. Since viscosity is neglected, this benchmark focuses on the inviscid aerodynamic response induced by geometry-dependent compressibility effects, including strong near-wall pressure variation and possible shock structures in the transonic regime.

\paragraph{Mesh and channels.}
Each sample is stored in a structured cylindrical view of resolution $221 \times 51$, corresponding to $11{,}271$ grid nodes per sample. The two coordinate arrays provide the physical mesh locations $(x,y)$, and the raw flow array contains five channels,
\[
    [\rho, u, v, p, \mathrm{Mach}] .
\]
The original Geo-FNO airfoil setting uses the physical coordinates as input and predicts the Mach number field. In our main 2D setting, we use the same airfoil geometries and mesh representation but predict the three-channel target $(u,v,p)$, so that both velocity and pressure errors are evaluated. Per-node inputs consist of the physical coordinates $(x,y)\in\mathbb{R}^2$, augmented with the offline geometry descriptors
\[
    G(x)
    =
    \bigl(
    \phi,\,
    \mathrm{window},\,
    \mathrm{conf},\,
    \mathbf{t},\,
    \mathbf{n},\,
    \kappa
    \bigr),
\]
where $\phi$ denotes the signed-distance value, $\mathrm{window}$ is the boundary-localization weight, $\mathrm{conf}$ is the confidence score of the local boundary association, $\mathbf{t}$ and $\mathbf{n}$ are the local tangent and normal directions, and $\kappa$ is the curvature descriptor.

For boundary-specific evaluation, we distinguish the physical airfoil wall from the artificial wake cut induced by the structured grid representation. In our implementation, the airfoil body surface is identified as the $121$-point arc on the inner $j{=}0$ row, corresponding to the Python slice $[50{:}171]$, i.e., column indices $50,\ldots,170$. The remaining nodes on the same inner row belong to the wake-cut portion of the grid and are therefore excluded from wall-based metrics and boundary-surface evaluations.

\paragraph{Split.}
We use the standard upstream split with $1000$ training samples and $200$ test samples.

\paragraph{Visualization.}
\Cref{fig:airfoil_tasks} visualizes representative airfoil geometries from the 2D benchmark.
The samples exhibit substantial geometric diversity in both thickness and camber, ranging from thin and nearly symmetric profiles to thicker and more strongly cambered shapes.
Here, $t$ denotes the normalized maximum thickness of each airfoil profile, and $c$ denotes the normalized maximum camber.
These geometric variations change the local curvature, leading-edge shape, trailing-edge behavior, and the spatial distribution of wall-normal and tangential directions along the airfoil surface.
As a result, the corresponding flow fields are not only geometry-dependent at the global level, but also highly sensitive to local near-wall structures.
This makes the benchmark suitable for evaluating whether a neural operator can adapt to different airfoil geometries and whether the proposed geometry-conditioned correction can improve predictions in regions where aerodynamic quantities are strongly affected by boundary shape.
\begin{figure}[htbp]
	\centering
\includegraphics[width=1.0\linewidth]{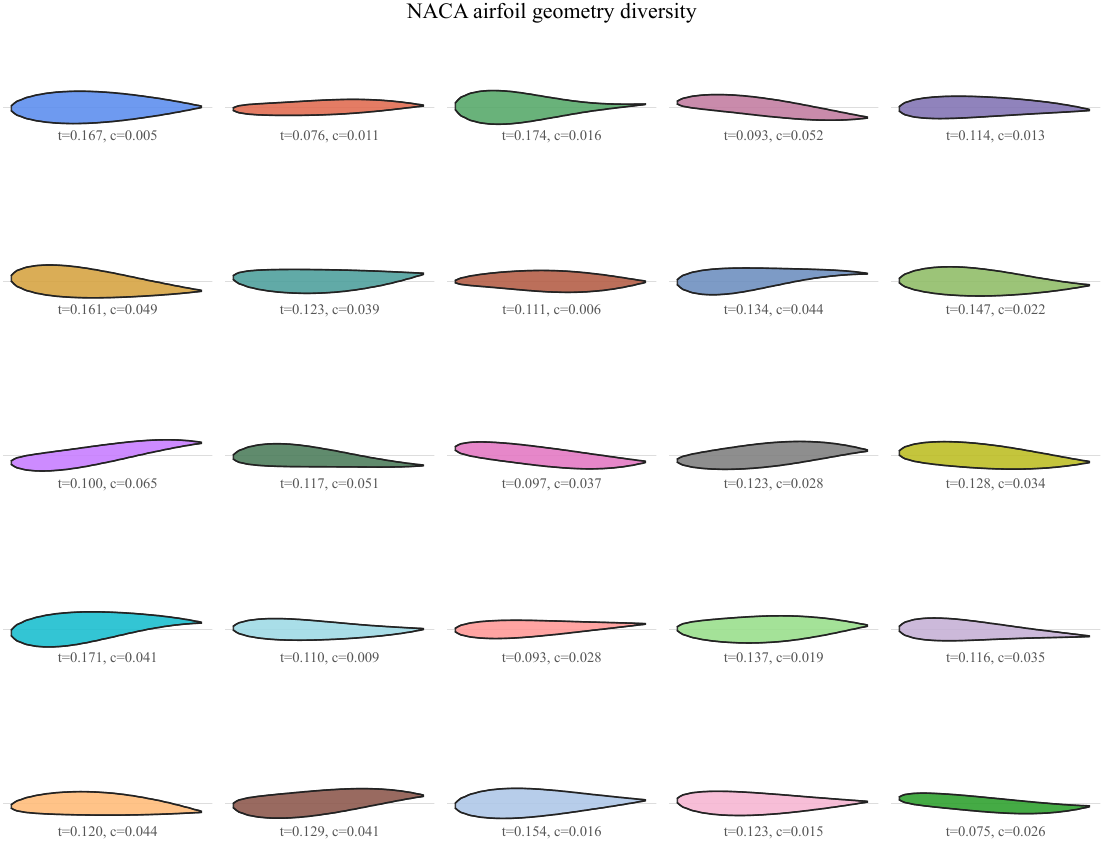}
	\caption{Visualization of 2D airfoil benchmark. Here, $t$ denotes the normalized maximum thickness, and $c$ denotes the normalized maximum camber of each airfoil profile.
}
	\label{fig:airfoil_tasks}
\end{figure}
    
    \subsection{ShapeNetCar 3D benchmark}
    \label{sec:appendix:dataset:car}
    
    \paragraph{Governing PDE.}
    The 3D benchmark targets steady external aerodynamic flow around generic vehicle bodies. 
    Following the ShapeNetCar setting used in Transolver, the underlying simulations correspond to incompressible Reynolds-averaged Navier--Stokes (RANS) flow around cars under a fixed driving condition of approximately $72$ km/h. 
    In conservative form, the incompressible mean-flow equations can be written as
    \begin{equation*}
        \nabla \!\cdot\! \mathbf{u} = 0,
        \qquad
        \rho(\mathbf{u}\!\cdot\!\nabla)\mathbf{u}
        =
        -\nabla p
        +
        \nabla\!\cdot\!\boldsymbol{\tau}_{\mathrm{eff}},
    \end{equation*}
    where $\mathbf{u}=(u,v,w)$ is the mean velocity field, $p$ is the pressure, $\rho$ is the fluid density, and $\boldsymbol{\tau}_{\mathrm{eff}}$ denotes the effective viscous stress including the modeled turbulent contribution. 
    This notation avoids committing to a specific turbulence closure, which is not exposed as part of the learning benchmark. 
    The far-field/inlet condition is fixed across the dataset, while a no-slip wall condition is imposed on the wetted vehicle surface. 
    The main engineering quantity of interest is the drag coefficient $C_D$, which is computed from the predicted aerodynamic fields, with the surface pressure contribution evaluated on the vehicle surface.
    
    \paragraph{Source and generation.}
    We follow the ShapeNetCar benchmark setting used by \citep{wu2024transolver}. 
    The original car geometries come from the ShapeNet ``car'' category and the CFD data are based on the ShapeNetCar dataset introduced by \citep{umetani2018learning}. 
    The benchmark contains $889$ vehicle geometries simulated under the same driving-speed condition. 
    Following the preprocessing protocol adopted in Transolver and related 3D geometry-conditioned surrogate models, each case is represented as an unstructured geometry with both surrounding-flow samples and wetted-surface samples. 
The task is to predict the surrounding velocity field together with the surface pressure field from the vehicle geometry.

\paragraph{Mesh and channels.}
Each sample is represented on an unstructured hybrid point set with $32{,}186$ mesh points. 
The point set combines volume points in the surrounding flow region and surface points on the wetted vehicle body. 
For each point, the base geometric input consists of the 3D position, signed distance to the vehicle surface, and an associated normal direction,
\[
    x_{\mathrm{in}}
    =
    (x,y,z,\phi,n_x,n_y,n_z)
    \in \mathbb{R}^{7}.
\]
Here, $\phi$ denotes the signed-distance value and $(n_x,n_y,n_z)$ denotes the local unit normal direction associated with the nearest or corresponding surface geometry. 
In our model, these base geometric inputs are further augmented with the offline geometry-cache descriptors used throughout this work,
\[
    G(x)
    =
    \bigl(
    \phi,\,
    \mathrm{window},\,
    \mathrm{conf},\,
    \mathbf{t}_1,\,
    \mathbf{t}_2,\,
    \mathbf{n},\,
    \kappa
    \bigr),
\]
where $\mathrm{window}$ is the boundary-localization weight, $\mathrm{conf}$ is the confidence score of the local boundary association, $\mathbf{t}_1$ and $\mathbf{t}_2$ are local tangential directions, $\mathbf{n}$ is the outward normal direction, and $\kappa$ is the local curvature descriptor. 
A binary surface mask is used internally to separate volume and surface nodes for loss computation and evaluation, but it is not counted as part of the standard geometric input channel.

The prediction target contains the three velocity components on the surrounding flow points and pressure on the wetted-surface points. 
Following the benchmark convention, velocity errors are evaluated on the volume point set, while pressure-based metrics and drag-related evaluations are computed on the surface point set. 
When computing surface-pressure metrics or the pressure contribution to $C_D$, we restrict the evaluation to the wetted vehicle surface and exclude non-surface volume samples.

\paragraph{Split.}
We use $789$ training cars and $100$ test cars.

\paragraph{Visualization.}
\Cref{fig:car_tasks} visualizes representative vehicle geometries from the ShapeNetCar 3D benchmark. 
Compared with the 2D airfoil benchmark, this task contains substantially more complex three-dimensional geometry variations, including different body lengths, roof profiles, front and rear shapes, hood slopes, and local surface curvature patterns. 
Each vehicle geometry is represented by an unstructured surface/volume point set, and the visualization highlights the diversity of wetted vehicle surfaces on which pressure-based quantities and drag-related metrics are evaluated. 
Here, $C_D$ denotes the drag coefficient of each car under the prescribed flow condition. 
The variation of $C_D$ across geometries indicates that changes in vehicle shape lead to different aerodynamic responses, especially through the pressure distribution and surface-normal force contribution on the car body. 
These geometry-dependent aerodynamic differences make the benchmark suitable for evaluating whether a surrogate model can handle complex 3D solid boundaries and whether GeoABC can improve boundary-sensitive predictions on wetted surfaces.

\begin{figure}[htbp]
	\centering
\includegraphics[width=1.0\linewidth]{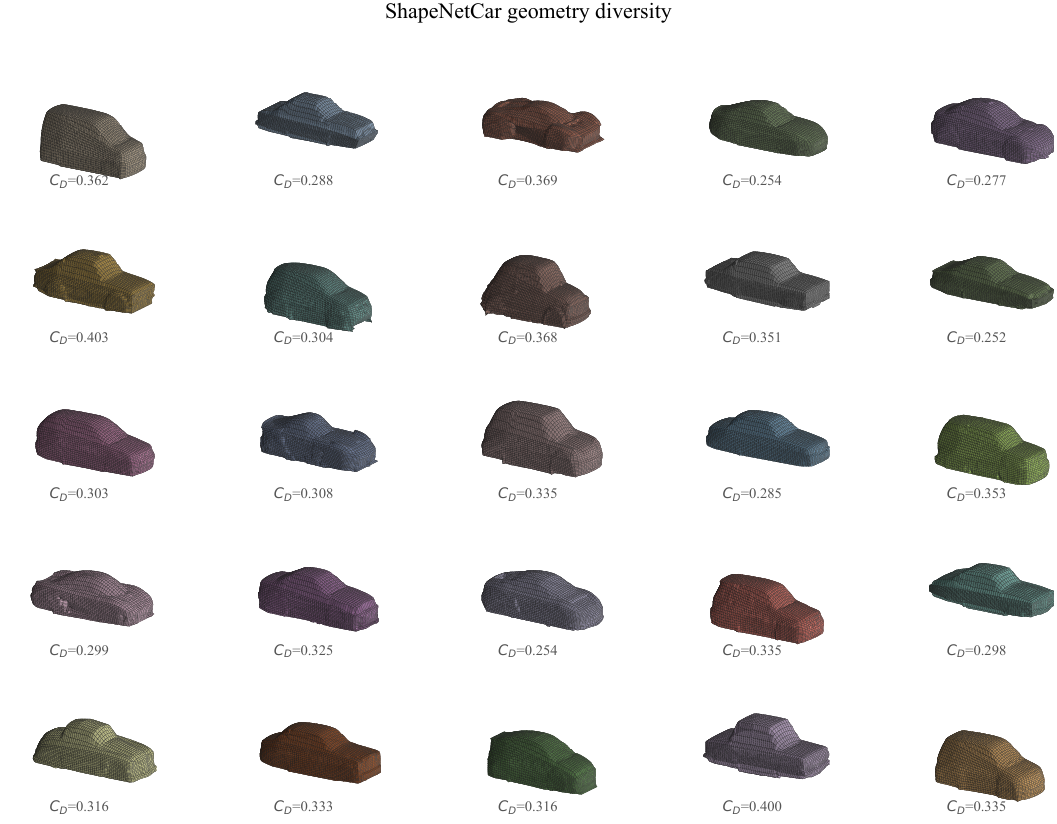}
	\caption{Visualization of 3D car benchmark. Here, $C_D$ denotes the drag coefficient of each car geometry under the prescribed flow condition.
}
	\label{fig:car_tasks}
\end{figure}
\FloatBarrier

\section{Additional Experiments}
\label{sec:appendix:additional_exp}
\subsection{Additional Diagnostics on Boundary-Aware Remedies}
\label{sec:appendix:boundary_awareness}
\begin{table}[htbp]
\centering
\caption{Comparison of standard boundary-aware remedies on the airfoil pressure-only validation set.}
\label{tab:motivation}
\resizebox{\linewidth}{!}{
\begin{tabular}{lcccc}
\toprule
Backbone & Global rel-$L_2$ $\downarrow$ & Near rel-$L_2$ $\downarrow$ & $C_p$ MAE $\downarrow$ & wall-$p$ MAE $\downarrow$\\
\midrule
\textbf{Geo-FNO} 
& $0.0717$
& $0.1780$
& $0.3055$
& $0.1369$ \\
\quad $+$ boundary loss reweighting
& $0.0721$
& $0.1772$
& $0.2922$
& $0.1309$\\
\quad $+$ scalar descriptor concatenation
& $0.0719$
& $0.1777$
& $0.3052$ 
& $0.1367$\\
\midrule
\textbf{GNOT}  
& $0.0271$           
& $0.0712$           
& $0.1106$          
& $0.0496$ \\
\quad $+$ cross-attention branch     
& $0.0438$ 
& $0.1058$
& $0.1512$
& $0.0678$\\
\bottomrule
\end{tabular}
}
\end{table}

\paragraph{Purpose.}
Before introducing GeoABC, we conduct a diagnostic study to test whether the near-wall accuracy gap can be closed by standard boundary-aware remedies.
The goal of this study is not to exhaust all possible boundary-enhanced architectures, but to isolate a common design assumption in existing remedies: the boundary is treated primarily as a proximity-defined region that should receive more input cues, larger training weights, or additional model capacity.
Such designs can improve the model's awareness of where the boundary is located, but they do not explicitly encode the local tangential--normal structure that organizes near-wall physical variation.

\paragraph{Diagnostic setting.}
We perform the diagnostic experiments on the NACA airfoil benchmark.
To make the boundary effect easier to interpret, we use the pressure-only prediction setting, where the output is the pressure field and the boundary-sensitive metrics can be directly evaluated through the pressure distribution on the airfoil surface.
We report global relative $L_2$, near-wall relative $L_2$, surface pressure-coefficient error $C_p$ MAE, and wall-pressure MAE.
The near-wall metric is computed inside the same narrow band around the solid boundary used throughout the paper.
All compared variants use the same data split, backbone configuration, optimizer setting, and training schedule as their corresponding baseline, so that the comparison focuses on how boundary information is incorporated.

\paragraph{Boundary-aware variants.}
We compare three representative boundary-aware remedies.

First, \emph{boundary loss reweighting} keeps the backbone architecture unchanged and increases the loss weight of samples inside the near-wall band.
This variant tests whether the near-wall gap can be reduced simply by assigning higher training priority to boundary-region errors.

Second, \emph{scalar descriptor concatenation} augments the model input with scalar geometric descriptors, including the signed distance to the boundary, the near-wall window value, and local curvature.
This variant tests whether explicitly exposing boundary proximity and local shape information to the model is sufficient.

Third, \emph{cross-attention boundary branch} introduces an additional boundary-specific branch that processes the same scalar descriptors and fuses the resulting boundary features back into the backbone representation through cross-attention.
This variant gives the model extra boundary-specific capacity, but the injected boundary information remains isotropic: it indicates which points are close to the wall, but does not define a local tangent--normal frame or direction-aware correction paths.

These three variants cover the common ways of improving boundary awareness: reweighting boundary errors, concatenating boundary descriptors, and adding boundary-specific capacity.
However, none of them explicitly models the directional structure of near-wall physics.
They therefore test whether boundary localization alone is sufficient, before introducing the anisotropic design of GeoABC.

\paragraph{Results.}
The results are summarized in \Cref{tab:motivation}.
On Geo-FNO, boundary loss reweighting and scalar descriptor concatenation barely change near-wall relative $L_2$.
The baseline near-wall error is $0.1780$, while reweighting gives $0.1772$ and descriptor concatenation gives $0.1777$.
Although reweighting slightly improves $C_p$ MAE and wall-pressure MAE, the near-wall field error remains essentially unchanged.
This suggests that simply emphasizing near-wall samples during training or exposing scalar boundary descriptors at the input level is insufficient to resolve the structural near-wall gap.

On the stronger GNOT backbone, the dedicated cross-attention boundary branch even degrades performance.
Near-wall relative $L_2$ increases from $0.0712$ to $0.1058$, while global relative $L_2$, $C_p$ MAE, and wall-pressure MAE all become worse.
This result is particularly important because the variant adds boundary-specific capacity rather than removing capacity.
Therefore, the failure cannot be explained by insufficient model size or lack of boundary information alone.
Instead, adding an isotropic boundary branch may interfere with the backbone representation when the added boundary signal is not aligned with the directional structure of near-wall physics.

\paragraph{Interpretation.}
These diagnostics show that boundary awareness alone is not enough.
The tested remedies can inform the model where the boundary is, increase the loss weight near the wall, or allocate additional capacity to boundary-region tokens.
However, they still treat the boundary as an isotropic high-error region.
They do not distinguish between tangential propagation along the wall and wall-normal variation across the boundary layer, nor do they model the coupling between these two directions.
As a result, they cannot consistently close the near-wall accuracy gap.

This observation motivates the design of GeoABC.
Rather than using boundary geometry only as scalar input features or training weights, GeoABC uses the solid boundary as a structural prior.
It constructs a local tangent--normal frame from geometry, decomposes boundary correction into tangential, normal, and cross-directional components, and writes the resulting correction into the intermediate representation of the backbone.
Thus, GeoABC addresses not only boundary localization, but also the anisotropic organization of near-wall physical responses.

\subsection{Spatial and Regional Attribution of GeoABC Gains}
\label{sec:appendix:spatial_attribution}

\paragraph{Attribution protocol.}
We analyze where the improvement of GeoABC comes from by decomposing the gain along surface position, wall distance, and near-wall versus non-near-wall regions.
For the 2D NACA airfoil benchmark, the near-wall band occupies only $5.42\%$ of grid points, while the non-near-wall region occupies the remaining $94.58\%$.
For the 3D car benchmark, the wetted-surface near-wall region occupies $11.10\%$ of points, while the non-near-wall region occupies $88.90\%$.
We therefore report two complementary quantities: ``regional gain'', which measures the percentage error reduction within each region, and ``gain contribution'', which measures the fraction of the total error reduction contributed by each region.
This distinction is important because a small near-wall band can have a large regional gain, while the much larger non-near-wall region can still contribute substantially to the total global error reduction.
The spatial attribution further tests whether the gain is a broad boundary-band effect rather than a leading-edge or high-curvature locality effect.
Together, these analyses examine whether GeoABC improves the targeted boundary region and whether its mid-layer correction propagates part of the benefit to the broader flow field.

\paragraph{Spatial attribution on 2D airfoil.}
As shown in \Cref{fig:spatial_attribution}, GeoABC substantially reduces pressure error along the airfoil surface.
The baseline exhibits large pointwise pressure-error peaks near chord locations with strong surface-pressure variation, whereas GeoABC keeps the error consistently low.
The distance-to-wall analysis further shows that the gain is concentrated near the wall and decays toward the far field.
At the farthest grid layer, the reduction becomes slightly negative, suggesting that GeoABC mainly improves the targeted boundary-relevant region rather than uniformly reducing error everywhere.
Finally, the curvature-quartile analysis shows that the reduction remains strong across low- and high-curvature boundary segments.
Thus, the improvement is not merely a leading-edge or high-curvature locality effect, but a broader boundary-band effect.

\begin{figure}[htbp]
    \centering
    \includegraphics[width=0.95\linewidth]{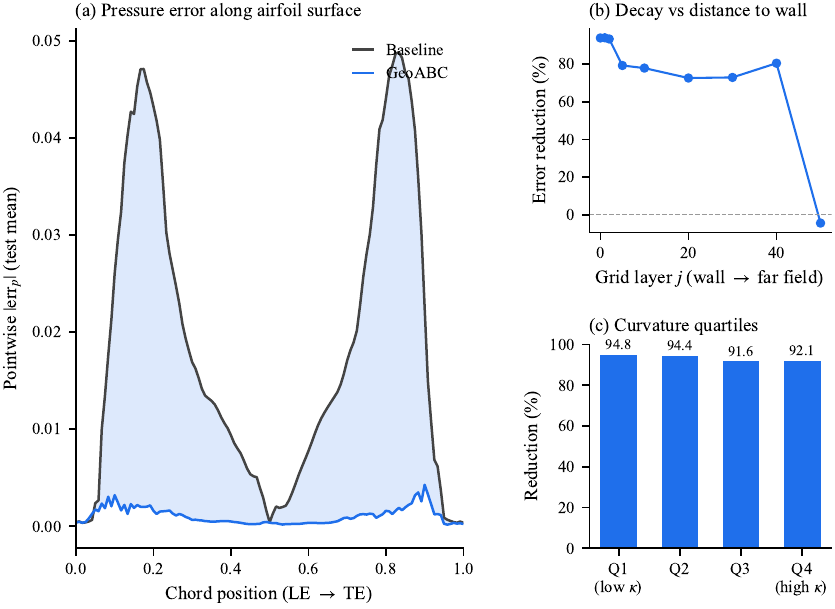}
    \caption{\textbf{Spatial attribution of GeoABC gains on Geo-FNO.}
    We analyze the multi-field $(u,v,p)$ NACA airfoil setting.
    (a) Test-mean pointwise pressure error along the airfoil surface; the shaded region denotes the per-position reduction over the baseline.
    (b) Error reduction as a function of grid layer distance from the wall to the far field.
    (c) Error reduction across boundary-curvature quartiles.
    GeoABC mainly improves the near-wall band, and the gain is broadly consistent across curvature regimes.}
    \label{fig:spatial_attribution}
\end{figure}

\paragraph{Regional decomposition across backbones.}
We next quantify how much of the global improvement is accumulated in the near-wall and non-near-wall regions.
For each region, we report two quantities.
``Gain contribution'' measures the fraction of the total error reduction contributed by that region, while ``regional gain'' measures the percentage error reduction within that region.
Therefore, the gain contribution reflects where the global improvement is accumulated, whereas the regional gain reflects how strongly GeoABC improves each region individually.

As shown in \Cref{tab:regional_gain_decomposition}, GeoABC produces strong near-wall regional gains across all 2D backbones.
Although the near-wall band occupies only $5.42\%$ of grid points, it contributes a substantial fraction of the total global MSE reduction.
For Geo-FNO, CNO, and GNOT, the near-wall regional gains reach $92.5\%$, $60.9\%$, and $71.0\%$, respectively.
At the same time, a non-trivial fraction of the total gain also comes from non-near-wall points, especially for Geo-FNO, GNOT, and Transolver.
This suggests that the mid-layer boundary correction does not remain an isolated local adjustment; after being written into the representation, it can propagate through the remaining backbone and improve the broader flow field.

\begin{table}[t]
\centering
\small
\caption{\textbf{Regional gain decomposition on 2D NACA $(u,v,p)$.}
Gain contribution reports the fraction of total global MSE reduction from each region, while regional gain reports the percentage MSE reduction within that region.}
\label{tab:regional_gain_decomposition}
\setlength{\tabcolsep}{5pt}
\begin{tabular}{lcccc}
\toprule
\multirow{2}{*}{Backbone}
& \multicolumn{2}{c}{Gain contribution}
& \multicolumn{2}{c}{Regional gain} \\
\cmidrule(lr){2-3}
\cmidrule(lr){4-5}
& Near & Non-near & Near & Non-near \\
\midrule
Geo-FNO    & $31.1\%$ & $68.9\%$ & $92.5\%$ & $71.9\%$ \\
CNO        & $62.3\%$ & $37.7\%$ & $60.9\%$ & $11.1\%$ \\
GNOT       & $33.8\%$ & $66.2\%$ & $71.0\%$ & $46.7\%$ \\
Transolver & $28.3\%$ & $71.7\%$ & $24.5\%$ & $20.8\%$ \\
\bottomrule
\end{tabular}
\end{table}

A similar pattern appears on the 3D car benchmark.
As shown in \Cref{tab:car3d_regional_gain_decomposition}, GeoABC consistently improves the wetted-surface near-wall region across all 3D backbones.
The near-wall regional gain is especially large on GNOT and Transolver, reaching $75.8\%$ and $77.1\%$, respectively.
For Geo-FNO3D, the near-wall region contributes $65.0\%$ of the total velocity-error reduction despite occupying only $11.10\%$ of points, indicating that the improvement is strongly concentrated on the wetted surface.
For GINO, GNOT, and Transolver, the non-near-wall region also contributes a large fraction of the total gain, showing that boundary-local correction can further benefit the surrounding velocity field after interacting with the remaining operator layers.

\begin{table}[htbp]
\centering
\small
\caption{\textbf{Regional gain decomposition on 3D TransolverCar.}
The decomposition uses squared rel-$L_2$ quantities corresponding to the volume-velocity and surface-velocity metrics in \Cref{tab:car3d}.}
\label{tab:car3d_regional_gain_decomposition}
\setlength{\tabcolsep}{5pt}
\begin{tabular}{lcccc}
\toprule
\multirow{2}{*}{Backbone}
& \multicolumn{2}{c}{Gain contribution}
& \multicolumn{2}{c}{Regional gain} \\
\cmidrule(lr){2-3}
\cmidrule(lr){4-5}
& Near & Non-near & Near & Non-near \\
\midrule
Geo-FNO3D  & $65.0\%$ & $35.0\%$ & $21.1\%$ & $2.3\%$ \\
GINO       & $33.5\%$ & $66.5\%$ & $31.5\%$ & $12.5\%$ \\
GNOT       & $34.8\%$ & $65.2\%$ & $75.8\%$ & $51.1\%$ \\
Transolver & $32.6\%$ & $67.4\%$ & $77.1\%$ & $54.5\%$ \\
\bottomrule
\end{tabular}
\end{table}

\paragraph{Takeaway.}
These analyses show that GeoABC primarily improves the targeted near-wall region, but its effect is not restricted to the boundary band.
Because GeoABC is inserted into an intermediate representation, the corrected near-wall features can continue to interact with the remaining operator layers.
As a result, the method reduces boundary-sensitive errors while also propagating part of the improvement to non-near-wall regions.
This supports our design choice of applying geometry-conditioned boundary correction in the latent space rather than as a purely output-level local post-processing step.
\FloatBarrier

\subsection{Additional Results on Pipe NS}
\label{sec:appendix:pipe}

\paragraph{Dataset.}
The Pipe NS dataset~\citep{li2023geofno} contains 2D incompressible Navier--Stokes flows in a straight channel with no-slip conditions on the top and bottom walls.
Each sample is defined on a uniform Cartesian grid of $129 \times 129$ nodes, and the target field contains the two velocity components and pressure, $(u,v,p)$.
We use the same $1000/200/200$ train/validation/test split as the main 2D experiments.

Pipe NS provides a complementary setting to the NACA airfoil benchmark.
Unlike the airfoil C-grid, whose near-wall region follows a curved solid boundary, the pipe domain contains two straight no-slip walls at $y=0$ and $y=1$.
Accordingly, the geometry cache is constructed with a \emph{dual-wall} boundary band: for each grid node, the boundary proximity is computed with respect to the closer of the two channel walls, and the near-wall band covers nodes within a fixed $\phi$-percentile of either wall.
This setting allows us to test whether GeoABC transfers beyond curved airfoil geometries to canonical wall-bounded flows.

\paragraph{Experimental setup.}
We apply GeoABC to two representative 2D backbones, Geo-FNO and CNO.
For a controlled comparison, we use the same injection locations, loss weights $(\lambda_b = 2.0, \lambda_a = 0.1)$, and 50-epoch training protocol as in the main NACA experiments.
Each configuration is trained once, and we report point estimates.
\Cref{tab:pipe} summarizes the results.

\begin{table}[t]
\centering
\caption{\textbf{GeoABC on Pipe NS.}
Percentage change is reported relative to the corresponding backbone baseline.}
\label{tab:pipe}
\small
\begin{tabular}{llccc}
\toprule
Backbone & Method & Global rel-$L_2$ & Near-wall rel-$L_2$ & Wall-$p$ MAE \\
\midrule
\multirow{2}{*}{Geo-FNO}
& Baseline & 0.328 & 0.839 & 0.230 \\
& GeoABC 
& 0.318 \textcolor{teal}{($-$3.0\%)}
& 0.474 \textcolor{teal}{($-$43.5\%)}
& 0.079 \textcolor{teal}{($-$65.8\%)} \\
\addlinespace
\multirow{2}{*}{CNO}
& Baseline & 0.319 & 0.765 & 0.224 \\
& GeoABC 
& 0.303 \textcolor{teal}{($-$5.0\%)}
& 0.606 \textcolor{teal}{($-$20.8\%)}
& 0.191 \textcolor{teal}{($-$14.9\%)} \\
\bottomrule
\end{tabular}
\end{table}

\paragraph{Analysis.}
GeoABC consistently improves boundary-sensitive metrics on Pipe NS.
On Geo-FNO, the near-wall rel-$L_2$ decreases from $0.839$ to $0.474$, corresponding to a $43.5\%$ reduction, while wall-pressure MAE decreases from $0.230$ to $0.079$, corresponding to a $65.8\%$ reduction.
On CNO, GeoABC also reduces near-wall rel-$L_2$ by $20.8\%$ and wall-pressure MAE by $14.9\%$.
These improvements show that the proposed correction is not limited to curved airfoil boundaries: it also transfers to straight no-slip walls, where the relevant physical structure is still organized by tangential and wall-normal directions.

The global rel-$L_2$ improvements are more moderate, with reductions of $3.0\%$ on Geo-FNO and $5.0\%$ on CNO.
This behavior is consistent with the design of GeoABC.
The correction is activated only inside the dual-wall boundary band, so its primary effect is expected to appear in near-wall velocity and pressure errors rather than uniformly across the entire channel interior.
Therefore, the results provide additional cross-dataset evidence that geometry-conditioned anisotropic boundary correction improves wall-bounded prediction while preserving the global-field behavior of the original backbone.

\section{Additional Ablation Study}
\label{appendix:add:ablation}
\begin{table}[htbp]
\centering
\scriptsize
\caption{\textbf{Component-isolation ablation on Geo-FNO.}
We report global relative $L_2$, near-wall relative $L_2$, near-wall velocity relative $L_2$, and body-pressure MAE on the multi-field $(u,v,p)$ airfoil task.
Results are reported as mean $\pm$ standard deviation.}
\label{tab:isolation_geofno}
\setlength{\tabcolsep}{5pt}
\renewcommand{\arraystretch}{1.08}
\resizebox{\linewidth}{!}{
\begin{tabular}{lcccc}
\toprule
Configuration & Global $\downarrow$ & Near $\downarrow$ & UV Near $\downarrow$ & Body Pres $\downarrow$ \\
\midrule
Geo-FNO
& $0.0608{\pm}0.0028$ & $0.1320{\pm}0.0061$ & $0.1307{\pm}0.0064$ & $0.0934{\pm}0.0040$ \\
\texttt{no-inject}
& $0.0618{\pm}0.0035$ & $0.1285{\pm}0.0067$ & $0.1274{\pm}0.0068$ & $0.0904{\pm}0.0058$ \\
\texttt{no-frame}
& $0.0495{\pm}0.0012$ & $0.0794{\pm}0.0055$ & $0.0779{\pm}0.0052$ & $0.0444{\pm}0.0027$ \\
\texttt{normal-only}
& $0.0294{\pm}0.0015$ & $0.0361{\pm}0.0009$ & $0.0361{\pm}0.0010$ & $0.0191{\pm}0.0001$ \\
\texttt{zero-geometry}
& $0.0511{\pm}0.0028$ & $0.0845{\pm}0.0051$ & $0.0828{\pm}0.0049$ & $0.0495{\pm}0.0040$ \\
\midrule
\textbf{GeoABC}
& $\mathbf{0.0297{\pm}0.0006}$ & $\mathbf{0.0369{\pm}0.0012}$ & $\mathbf{0.0369{\pm}0.0007}$ & $\mathbf{0.0191{\pm}0.0017}$ \\
\bottomrule
\end{tabular}
}
\end{table}

\begin{table}[t]
\centering
\scriptsize
\caption{\textbf{Component-isolation ablation on CNO.}
We report global relative $L_2$, near-wall relative $L_2$, near-wall velocity relative $L_2$, and body-pressure MAE on the multi-field $(u,v,p)$ airfoil task.
Results are reported as mean $\pm$ standard deviation.}
\label{tab:isolation_cno}
\setlength{\tabcolsep}{5pt}
\renewcommand{\arraystretch}{1.08}
\resizebox{\linewidth}{!}{
\begin{tabular}{lcccc}
\toprule
Configuration & Global $\downarrow$ & Near $\downarrow$ & UV Near $\downarrow$ & Body Pres $\downarrow$ \\
\midrule
CNO
& $0.0194{\pm}0.0044$ & $0.0398{\pm}0.0094$ & $0.0409{\pm}0.0095$ & $0.0248{\pm}0.0061$ \\
\texttt{no-inject}
& $0.0315{\pm}0.0001$ & $0.0485{\pm}0.0049$ & $0.0486{\pm}0.0050$ & $0.0242{\pm}0.0001$ \\
\texttt{no-frame}
& $0.0295{\pm}0.0134$ & $0.0483{\pm}0.0274$ & $0.0484{\pm}0.0267$ & $0.0262{\pm}0.0154$ \\
\texttt{normal-only}
& $0.0190{\pm}0.0017$ & $0.0285{\pm}0.0023$ & $0.0289{\pm}0.0019$ & $0.0145{\pm}0.0026$ \\
\texttt{zero-geometry}
& $0.0270{\pm}0.0044$ & $0.0437{\pm}0.0094$ & $0.0441{\pm}0.0095$ & $0.0204{\pm}0.0061$ \\
\midrule
\textbf{GeoABC}
& $\mathbf{0.0171{\pm}0.0027}$ & $\mathbf{0.0242{\pm}0.0036}$ & $\mathbf{0.0246{\pm}0.0035}$ & $\mathbf{0.0121{\pm}0.0014}$ \\
\bottomrule
\end{tabular}
}
\end{table}

\begin{table}[t]
\centering
\scriptsize
\caption{\textbf{Component-isolation ablation on GINO Car.}
We report surface-velocity relative $L_2$, volume-velocity relative $L_2$, pressure relative $L_2$, and surface-pressure MAE on the 3D car benchmark.
Results are reported as mean $\pm$ standard deviation.}
\label{tab:isolation_gino}
\setlength{\tabcolsep}{5pt}
\renewcommand{\arraystretch}{1.08}
\resizebox{\linewidth}{!}{
\begin{tabular}{lcccc}
\toprule
Configuration & Surf. Vel $\downarrow$ & Vol. Vel $\downarrow$ & Pressure $\downarrow$ & Surf. Pres $\downarrow$ \\
\midrule
GINO
& $0.0432{\pm}0.0015$ & $0.0353{\pm}0.0008$ & $0.1139{\pm}0.0043$ & $3.2546{\pm}0.0993$ \\
\texttt{no-inject}
& $0.0391{\pm}0.0002$ & $0.0348{\pm}0.0001$ & $0.1107{\pm}0.0020$ & $3.1726{\pm}0.0558$ \\
\texttt{no-frame}
& $0.0380{\pm}0.0008$ & $0.0340{\pm}0.0004$ & $0.1114{\pm}0.0020$ & $3.1523{\pm}0.0526$ \\
\texttt{normal-only}
& $0.0382{\pm}0.0004$ & $0.0341{\pm}0.0002$ & $0.1111{\pm}0.0025$ & $3.1370{\pm}0.0287$ \\
\texttt{zero-geometry}
& $0.4693{\pm}0.0168$ & $0.3838{\pm}0.0091$ & $1.2387{\pm}0.0462$ & $35.3888{\pm}1.0797$ \\
\midrule
\textbf{GeoABC}
& $\mathbf{0.0378{\pm}0.0006}$ & $\mathbf{0.0340{\pm}0.0003}$ & $\mathbf{0.1102{\pm}0.0021}$ & $\mathbf{3.1057{\pm}0.0438}$ \\
\bottomrule
\end{tabular}
}
\end{table}

\subsection{Component isolation details.}
\label{sec:appendix:component_details}
\Cref{tab:isolation_geofno,tab:isolation_cno,tab:isolation_gino} provide detailed component-isolation ablations of GeoABC across representative 2D and 3D backbones.
Unlike the compact ablation discussion in the main text, these tables report multiple metrics for each backbone, allowing us to inspect whether each component affects only the near-wall metric or also changes global-field and surface-pressure behavior.

The \texttt{no-inject} variant removes the mid-layer correction write-back while keeping the same loss design.
It tests whether the improvement can be attributed to the boundary-related training losses alone, without modifying the backbone representation.

The \texttt{no-frame} variant keeps the correction module but removes the local tangent--normal frame.
Thus, the model still has additional correction capacity, but the correction is no longer organized by the boundary-aligned tangential and wall-normal directions.

The \texttt{normal-only} variant keeps the local frame but only retains the wall-normal branch.
It removes the tangential branch and the tangential--normal coupling branch, testing whether wall-normal correction alone is sufficient.

The \texttt{zero-geometry} variant keeps the full correction architecture but zeros out the local geometric inputs used by the geometry modulator.
Therefore, the module still has comparable correction capacity, but it can no longer condition the correction on local boundary proximity, curvature, and orientation.
This variant isolates the role of geometry-conditioned modulation from the role of additional parameters.

\paragraph{Boundary losses alone are not sufficient.}
The \texttt{no-inject} variant shows that the three-loss objective alone cannot explain the improvement of GeoABC.
On Geo-FNO, \texttt{no-inject} only changes near-wall relative $L_2$ from $0.1320$ to $0.1285$, a small reduction of $2.7\%$.
The corresponding UV near-wall error changes from $0.1307$ to $0.1274$, and body-pressure MAE changes from $0.0934$ to $0.0904$.
These changes are minor compared with full GeoABC, which reduces the same metrics to $0.0369$, $0.0369$, and $0.0191$, respectively.

The same conclusion is clearer on CNO.
With \texttt{no-inject}, the global error increases from $0.0194$ to $0.0315$, and near-wall error increases from $0.0398$ to $0.0485$.
Thus, merely adding boundary-related losses can even hurt the backbone if no representation-level correction is available.
On GINO Car, \texttt{no-inject} gives a moderate improvement in surface velocity from $0.0432$ to $0.0391$, but the gains in volume velocity, pressure, and surface-pressure MAE remain small.
Overall, \texttt{no-inject} confirms that the main improvement of GeoABC does not come from loss reweighting or auxiliary supervision alone.

\paragraph{The tangent--normal frame organizes boundary correction.}
The \texttt{no-frame} variant shows what happens when the model has a correction module but lacks the boundary-aligned coordinate system.
On Geo-FNO, \texttt{no-frame} reduces near-wall error from $0.1320$ to $0.0794$, which is a sizable improvement over the baseline.
However, full GeoABC further reduces it to $0.0369$, corresponding to an additional $53.5\%$ reduction relative to \texttt{no-frame}.
The same pattern appears in UV near-wall error, where \texttt{no-frame} reaches $0.0779$, while full GeoABC reaches $0.0369$.
This indicates that a generic correction module can help, but a boundary-aligned frame is needed to fully organize the correction.

On CNO, removing the frame is harmful.
The near-wall error increases from the baseline value $0.0398$ to $0.0483$, and the standard deviation is also large.
This suggests that an unstructured correction can interact unstably with a convolutional operator backbone.
On GINO Car, \texttt{no-frame} performs close to full GeoABC on surface velocity and volume velocity, but remains slightly worse on pressure and surface-pressure MAE.
Taken together, the frame is not merely an implementation detail: it provides a structural coordinate system that makes the correction more reliable across different operator families.

\paragraph{The wall-normal branch captures the dominant near-wall effect.}
The \texttt{normal-only} variant is consistently strong.
On Geo-FNO, it reduces global error from $0.0608$ to $0.0294$, near-wall error from $0.1320$ to $0.0361$, UV near-wall error from $0.1307$ to $0.0361$, and body-pressure MAE from $0.0934$ to $0.0191$.
These correspond to reductions of approximately $51.6\%$, $72.7\%$, $72.4\%$, and $79.6\%$, respectively.
This confirms that the wall-normal direction is a dominant source of near-wall error in the airfoil task.

On CNO, \texttt{normal-only} also improves all boundary-sensitive metrics.
Near-wall error decreases from $0.0398$ to $0.0285$, UV near-wall error decreases from $0.0409$ to $0.0289$, and body-pressure MAE decreases from $0.0248$ to $0.0145$.
However, full GeoABC further improves these values to $0.0242$, $0.0246$, and $0.0121$.
Thus, while wall-normal correction explains a large part of the gain, the tangential and cross-directional branches still provide additional benefit, especially on CNO.

On GINO Car, \texttt{normal-only} is also competitive, reducing surface velocity from $0.0432$ to $0.0382$ and surface-pressure MAE from $3.2546$ to $3.1370$.
Full GeoABC further improves these metrics to $0.0378$ and $3.1057$.
The improvement over \texttt{normal-only} is smaller in 3D than in 2D, but it is consistent across all four GINO metrics.
Therefore, wall-normal correction is the strongest single component, while the full anisotropic design provides the most robust overall profile across backbones and metrics.

\paragraph{Geometry-conditioned modulation is essential.}
The \texttt{zero-geometry} variant directly tests whether the gains come from extra correction capacity alone.
On Geo-FNO, \texttt{zero-geometry} improves over the baseline, reducing near-wall error from $0.1320$ to $0.0845$ and body-pressure MAE from $0.0934$ to $0.0495$.
However, it remains far worse than full GeoABC, which reaches $0.0369$ near-wall error and $0.0191$ body-pressure MAE.
Relative to \texttt{zero-geometry}, full GeoABC further reduces near-wall error by about $56.3\%$ and body-pressure MAE by about $61.4\%$.
This shows that additional correction capacity alone is not sufficient; the correction must be conditioned on local geometric information.

On CNO, \texttt{zero-geometry} is also weaker than full GeoABC.
It increases global error from $0.0194$ to $0.0270$ and near-wall error from $0.0398$ to $0.0437$.
Although body-pressure MAE decreases from $0.0248$ to $0.0204$, full GeoABC is still much better, reaching $0.0121$.
Thus, without geometry-conditioned modulation, the correction does not consistently improve the field-level errors.

The effect is most pronounced on GINO Car.
When local geometric inputs are removed, all metrics collapse: surface velocity increases from $0.0432$ to $0.4693$, volume velocity from $0.0353$ to $0.3838$, pressure from $0.1139$ to $1.2387$, and surface-pressure MAE from $3.2546$ to $35.3888$.
These errors are roughly one order of magnitude larger than the baseline.
This sharp degradation indicates that, in the 3D surface setting, an active correction module without reliable geometric conditioning can severely distort the learned representation.
Therefore, local geometry is not a minor auxiliary feature; it determines where the correction should act, how strongly it should be applied, and how it should be oriented.

\paragraph{Cross-metric interpretation.}
Across the three tables, GeoABC improves not only the selected near-wall metric but also multiple boundary-sensitive quantities.
On Geo-FNO, full GeoABC reduces global error from $0.0608$ to $0.0297$, near-wall error from $0.1320$ to $0.0369$, UV near-wall error from $0.1307$ to $0.0369$, and body-pressure MAE from $0.0934$ to $0.0191$.
The largest reductions appear in near-wall and body-pressure metrics, which are consistent with the design goal of GeoABC.

On CNO, full GeoABC reduces global error from $0.0194$ to $0.0171$, near-wall error from $0.0398$ to $0.0242$, UV near-wall error from $0.0409$ to $0.0246$, and body-pressure MAE from $0.0248$ to $0.0121$.
The global improvement is moderate, while the boundary-sensitive improvements are much larger.
This pattern supports the claim that GeoABC mainly targets near-wall errors rather than simply improving the whole field uniformly.

On GINO Car, full GeoABC reduces surface velocity from $0.0432$ to $0.0378$, volume velocity from $0.0353$ to $0.0340$, pressure from $0.1139$ to $0.1102$, and surface-pressure MAE from $3.2546$ to $3.1057$.
The gains are smaller than in the 2D airfoil setting, but they are consistent across surface, volume, pressure, and surface-pressure metrics.
This indicates that the same anisotropic correction principle remains beneficial in 3D, even though the stronger geometric and physical complexity makes the margin more moderate.

\paragraph{Overall conclusion.}
These ablations support four conclusions.
First, boundary-related losses alone are insufficient, as shown by the weak and inconsistent behavior of \texttt{no-inject}.
Second, the correction module must be organized by the local tangent--normal frame; otherwise, the added correction capacity can be suboptimal or unstable.
Third, the wall-normal branch captures the dominant near-wall effect, but the full tangential--normal--cross design gives the best cross-backbone behavior overall.
Fourth, geometry-conditioned modulation is essential: removing local geometric inputs substantially weakens the method and can cause severe failure in 3D.
Together, these results show that GeoABC's gains come from the combination of mid-layer correction, local boundary-frame organization, direction-aware branches, and geometry-conditioned modulation, rather than from extra parameters, boundary-aware losses, or boundary localization alone.

\subsection{Ablation on Injection Position}
\label{sec:appendix:injection_position}
\begin{table}[t]
\centering
\small
\setlength{\tabcolsep}{5.5pt}
\caption{\textbf{Ablation on the injection position.}
Relative reductions are computed against the baseline without injection.}
\label{tab:injection_position_ablation}
\resizebox{\linewidth}{!}{
\begin{tabular}{lcccc}
\toprule
Injection position & Global rel-$L_2$ $\downarrow$ & Near-wall rel-$L_2$ $\downarrow$ & $C_p$ MAE $\downarrow$ & Body Pres. MAE $\downarrow$ \\
\midrule
Baseline w/o injection
& 0.07170
& 0.17800
& 0.30549
& 0.13686 \\
\midrule
After \texttt{conv1} (early)
& 0.06504 {\scriptsize (-9.3\%)}
& 0.11450 {\scriptsize (-35.7\%)}
& 0.15940 {\scriptsize (-47.8\%)}
& 0.07139 {\scriptsize (-47.8\%)} \\

After \texttt{conv2} (middle)
& \textbf{0.03605} {\scriptsize (-49.7\%)}
& \textbf{0.05615} {\scriptsize (-68.5\%)}
& \textbf{0.06182} {\scriptsize (-79.8\%)}
& \textbf{0.02769} {\scriptsize (-79.8\%)} \\

After \texttt{conv3} (late)
& 0.06944 {\scriptsize (-3.2\%)}
& 0.14790 {\scriptsize (-16.9\%)}
& 0.23760 {\scriptsize (-22.2\%)}
& 0.10645 {\scriptsize (-22.2\%)} \\
\bottomrule
\end{tabular}
}
\end{table}
We conduct an injection-position ablation to examine where the GeoABC correction module should be inserted into the Geo-FNO backbone.
The sweep is performed under the pressure-only NACA airfoil body-only setting.
All variants are trained for 50 epochs, and we report the validation-best checkpoint from a single seed.
The baseline is the same Geo-FNO backbone without the GeoABC correction module.
For each injected variant, the relative reduction in parentheses is computed against this no-injection baseline.

As shown in \Cref{tab:injection_position_ablation}, inserting GeoABC after \texttt{conv2} achieves the best performance across all metrics.
It reduces global rel-$L_2$ from $0.07170$ to $0.03605$, near-wall rel-$L_2$ from $0.17800$ to $0.05615$, $C_p$ MAE from $0.30549$ to $0.06182$, and body-pressure MAE from $0.13686$ to $0.02769$.
This confirms that the proposed correction is most effective when applied to an intermediate representation.

Early insertion after \texttt{conv1} already improves the boundary-sensitive metrics, but it is clearly weaker than the middle-layer insertion.
This suggests that very early features have not yet accumulated sufficient global physical context, so the geometry-conditioned correction is applied before the backbone has formed a sufficiently informative flow representation.
By contrast, late insertion after \texttt{conv3} also underperforms because the correction is introduced too close to the output head.
At this stage, there are too few subsequent operator layers for the corrected boundary representation to interact with the global field and propagate through the decoder.
Therefore, the injection-position ablation supports our mid-layer design: GeoABC should operate after the backbone has built useful physical context, but before the representation becomes too close to final decoding.

\subsection{Boundary-Loss Weight Ablation and Band-Definition Sensitivity}
\label{sec:appendix:lambda_b_band_sensitivity}

We include an additional controlled study in \cref{tab:robustness_lambda_b}
to examine whether GeoABC depends on a finely tuned boundary-loss weight or on
a particular definition of the near-wall band. The experiment uses the Geo-FNO
backbone on the NACA airfoil benchmark with body-only boundary geometry, trains
the pressure field for $50$ epochs, and reports the validation-best checkpoint.
In the top block, we fix $\lambda_a{=}0.1$ and sweep
$\lambda_b\in\{1,2,4\}$. The results show a smooth tradeoff rather than
hyperparameter brittleness: $\lambda_b{=}1$ gives the lowest global rel-$L_2$,
whereas larger $\lambda_b$ increasingly favors boundary-local quantities.
Moving from $\lambda_b{=}1$ to $\lambda_b{=}4$ improves near rel-$L_2$ from
$0.06049$ to $0.05235$, $C_p$ MAE from $0.07143$ to $0.06068$, and body
pressure MAE from $0.03200$ to $0.02718$, while the global rel-$L_2$ changes
only mildly from $0.03420$ to $0.03627$. All three settings remain far better
than the baseline, whose near rel-$L_2$ is $0.17800$, indicating that the
boundary correction is not tied to a narrow choice of $\lambda_b$.

The bottom block further checks whether the reported near-wall improvement is
an artefact of the cached boundary window. We re-evaluate near rel-$L_2$ on
$\phi$-top-$k\%$ bands, where smaller $k$ selects points closer to the airfoil
surface. GeoABC consistently outperforms the baseline for all tested bands.
The absolute gap is largest in the tightest $3\%$ band
($0.19105\rightarrow0.05060$), and remains substantial as the band widens to
$20\%$ ($0.14460\rightarrow0.06667$). This pattern suggests that the gain is
strongest exactly where boundary errors are expected to concentrate, while
still extending beyond the thinnest wall-adjacent region. Therefore, the
near-wall improvement is not caused by one particular narrow-band definition;
it reflects a stable boundary-local correction effect.

\begin{table}[t]
\centering
\small
\caption{\textbf{Boundary-loss weight ablation and band-definition sensitivity.}
The $\lambda_b$ sweep shows a smooth global/near-wall tradeoff, while the
$\phi$-top-$k\%$ evaluation shows that GeoABC's near-wall advantage persists
across multiple boundary-band definitions.}
\label{tab:robustness_lambda_b}
\begin{tabular}{lcccc}
\toprule
\multicolumn{5}{l}{\textbf{(a)} $\lambda_b$ sweep with $\lambda_a{=}0.1$ fixed} \\
\midrule
Method & global rel-$L_2$ & near rel-$L_2$ & $C_p$ MAE & body $p$ MAE \\
\midrule
Baseline & 0.07170 & 0.17800 & 0.30549 & 0.13686 \\
GeoABC, $\lambda_b{=}1$ & \textbf{0.03420} & 0.06049 & 0.07143 & 0.03200 \\
GeoABC, $\lambda_b{=}2$ & 0.03605 & 0.05615 & 0.06182 & 0.02769 \\
GeoABC, $\lambda_b{=}4$ & 0.03627 & \textbf{0.05235} & \textbf{0.06068} & \textbf{0.02718} \\
\midrule
\multicolumn{5}{l}{\textbf{(b)} Near rel-$L_2$ under alternative boundary-band definitions} \\
\midrule
Method & $\phi$-top 3\% & $\phi$-top 5.4\% & $\phi$-top 10\% & $\phi$-top 20\% \\
\midrule
Baseline & 0.19105 & 0.18333 & 0.16785 & 0.14460 \\
GeoABC & \textbf{0.05060} & \textbf{0.05722} & \textbf{0.06855} & \textbf{0.06667} \\
\bottomrule
\end{tabular}
\end{table}

\section{Additional Visualizations}
\label{sec:appendix:visualization}
\begin{figure}[htbp]
    \centering
    \includegraphics[width=0.95\linewidth]{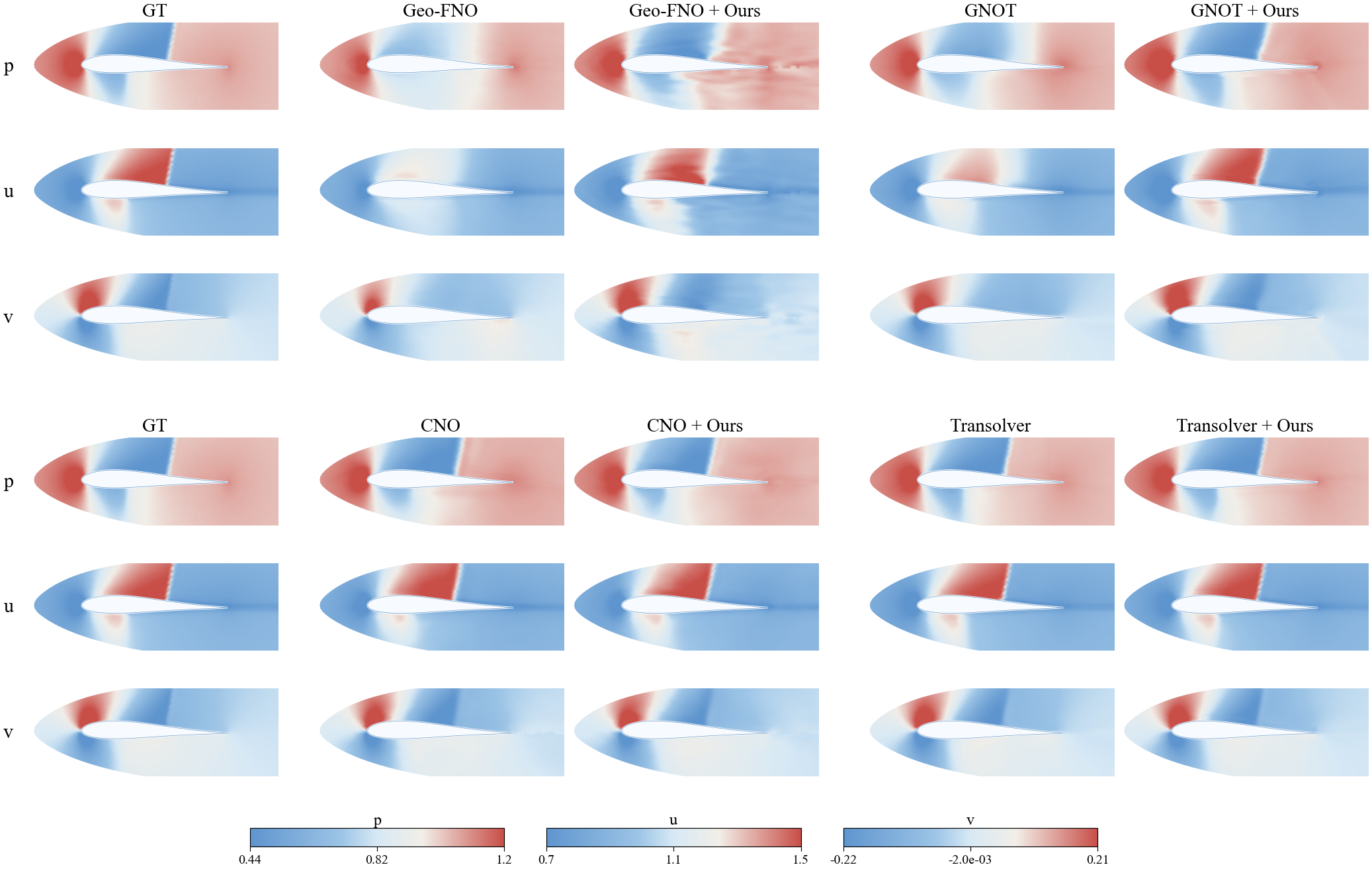}
    \caption{
Qualitative prediction comparison on the 2D NACA airfoil benchmark.
We visualize the ground truth fields and predictions from four backbones, with each baseline paired with its GeoABC-enhanced variant.
Rows show pressure $p$ and velocity components $u,v$.
GeoABC improves the near-wall and wake structures across Geo-FNO, GNOT, CNO, and Transolver, reducing the blurred or displaced patterns visible in the corresponding baselines.
}

    \label{fig:2d_pred_1}
\end{figure}

\begin{figure}[htbp]
    \centering
    \includegraphics[width=1.0\linewidth]{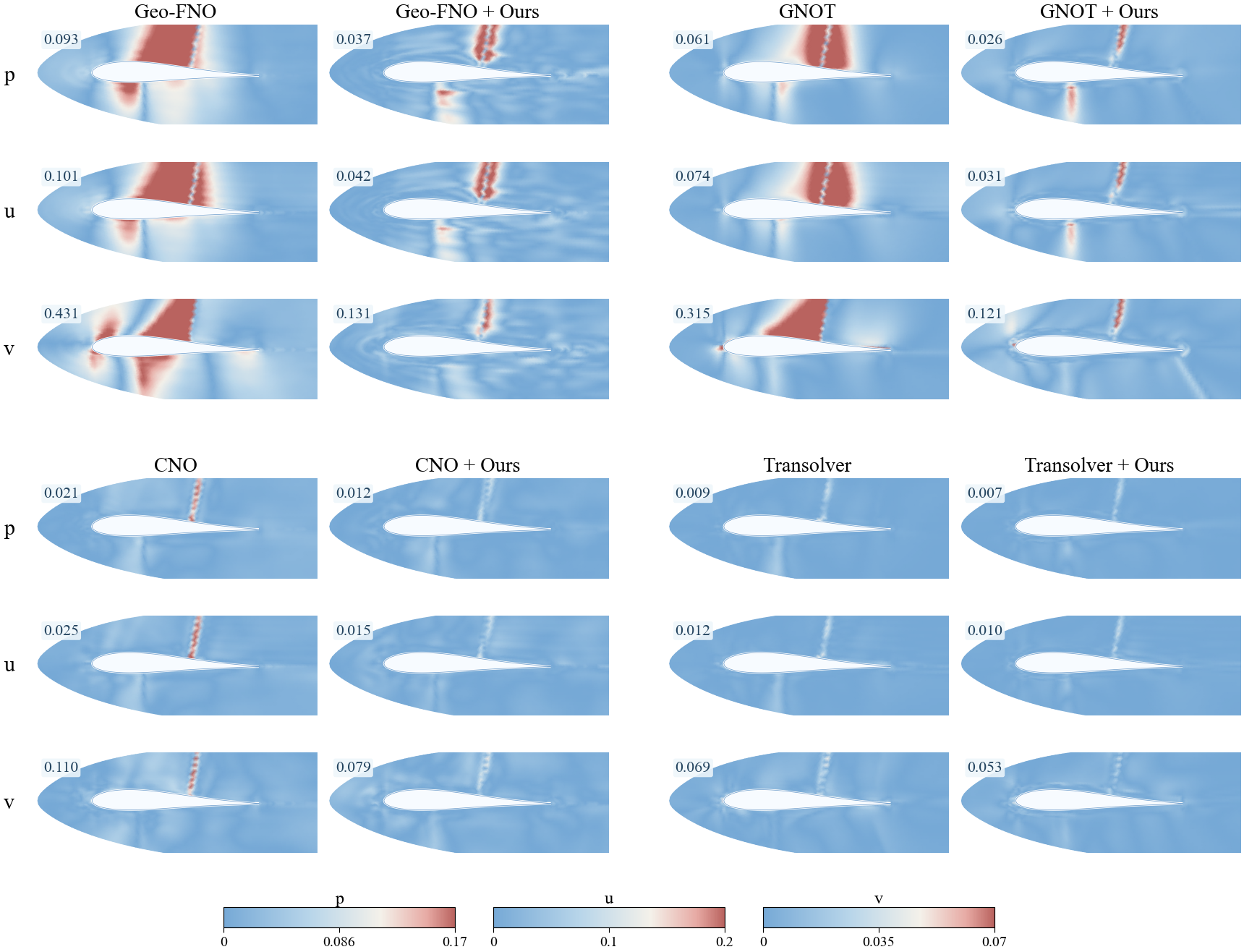}
    \caption{
Absolute error maps on the 2D NACA airfoil benchmark.
Each panel shows the pointwise absolute error between prediction and ground truth for pressure $p$ and velocity components $u,v$.
The overlaid numbers report near-wall relative $\ell_2$ error in the cropped airfoil region.
Across the four backbones, GeoABC consistently reduces concentrated near-wall and wake errors compared with the corresponding baselines.
}

    \label{fig:2d_err_1}
\end{figure}

\begin{figure}[htbp]
    \centering
    \includegraphics[width=1.0\linewidth]{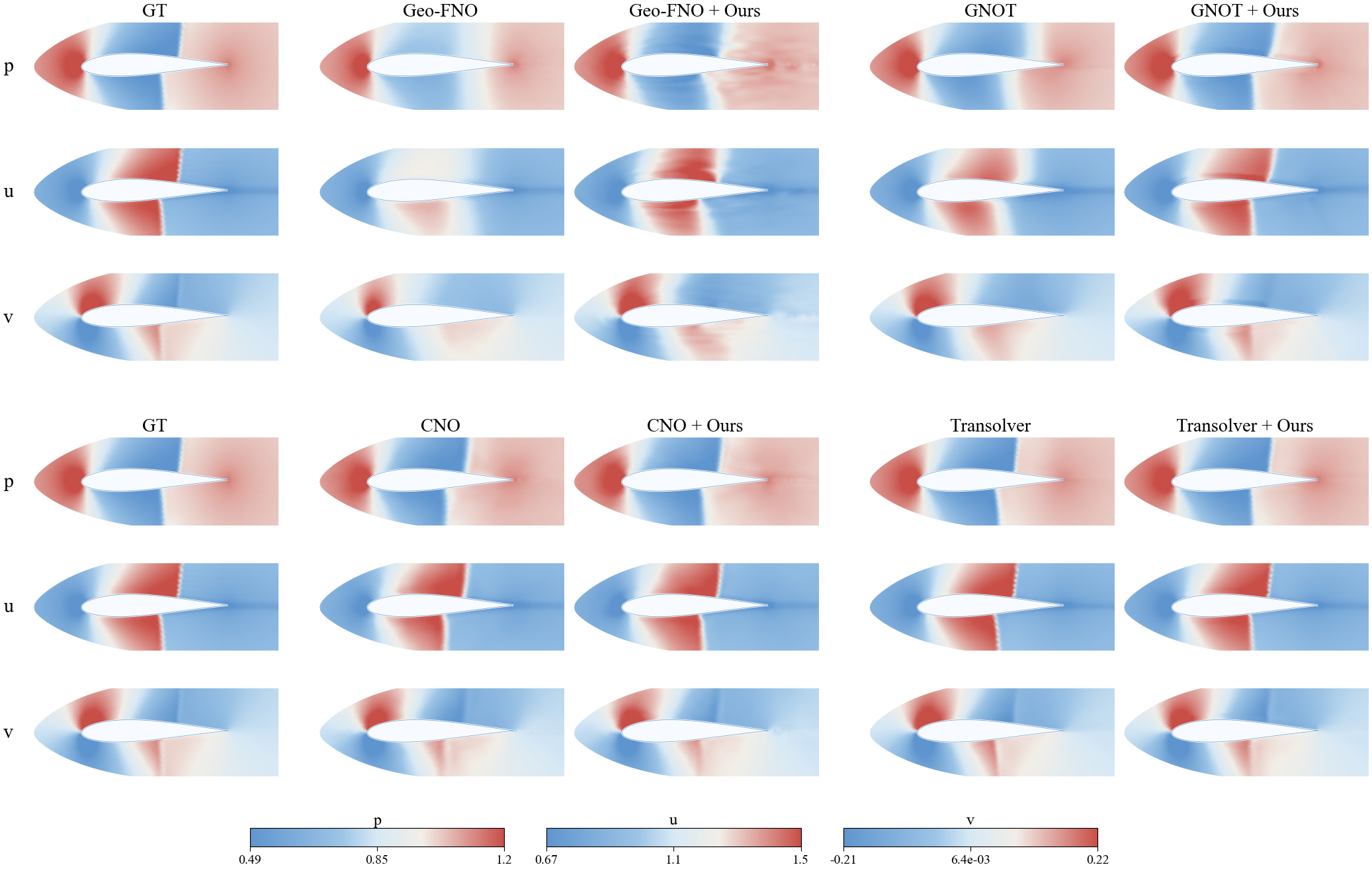}
    \caption{
Qualitative prediction comparison on the 2D NACA airfoil benchmark.
We visualize the ground truth fields and predictions from four backbones, with each baseline paired with its GeoABC-enhanced variant.
Rows show pressure $p$ and velocity components $u,v$.
GeoABC improves the near-wall and wake structures across Geo-FNO, GNOT, CNO, and Transolver, reducing the blurred or displaced patterns visible in the corresponding baselines.
}

    \label{fig:2d_pred_2}
\end{figure}

\begin{figure}[htbp]
    \centering
    \includegraphics[width=1.0\linewidth]{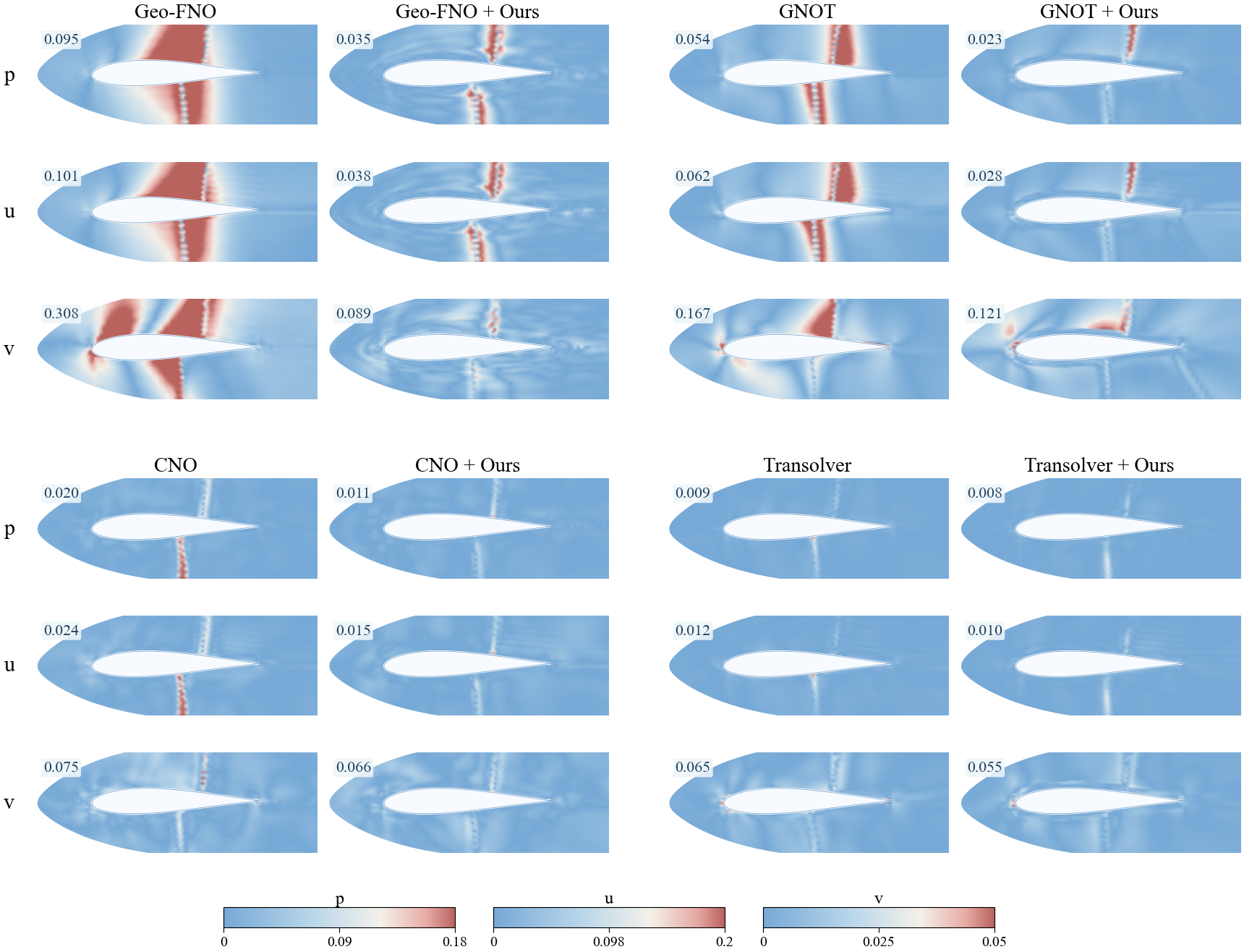}
    \caption{
Absolute error maps on the 2D NACA airfoil benchmark.
Each panel shows the pointwise absolute error between prediction and ground truth for pressure $p$ and velocity components $u,v$.
The overlaid numbers report near-wall relative $\ell_2$ error in the cropped airfoil region.
Across the four backbones, GeoABC consistently reduces concentrated near-wall and wake errors compared with the corresponding baselines.
}

    \label{fig:2d_err_2}
\end{figure}

\begin{figure}[htbp]
    \centering
    \includegraphics[width=1.0\linewidth]{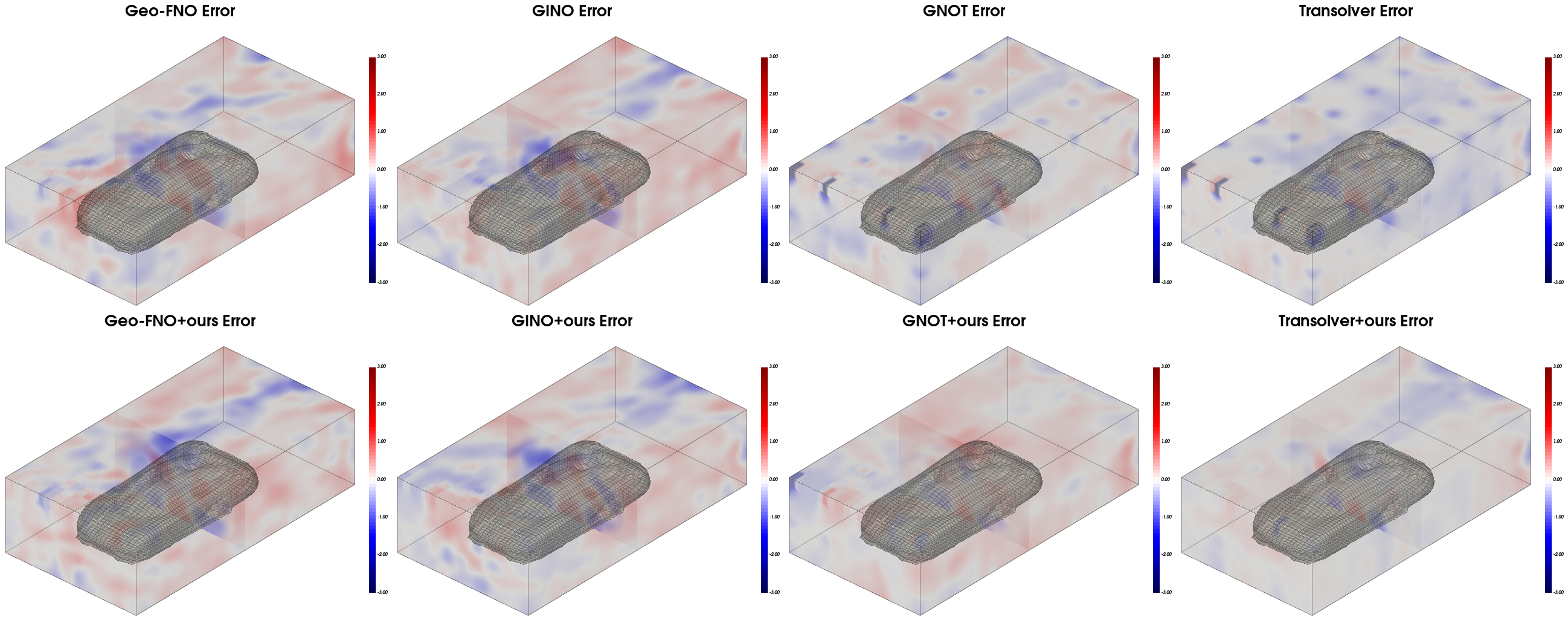}
\caption{
Visualization of surrounding velocity prediction errors on the Shape-Net Car benchmark.
Each panel shows the signed error of the predicted velocity magnitude in the 3D flow domain around the car, computed as $\|\hat{\mathbf{u}}\|_2 - \|\mathbf{u}\|_2$ and visualized on translucent slices of the surrounding volume.
The car surface is rendered in light gray with mesh edges for geometric reference.
The top row shows the baseline models, while the bottom row shows the corresponding variants enhanced with our method.
All panels share the same color scale, where red and blue indicate over- and under-predicted velocity magnitude, respectively.
}

    \label{fig:3dcar_v}
\end{figure}

\begin{figure}[htbp]
    \centering
    \includegraphics[width=1.0\linewidth]{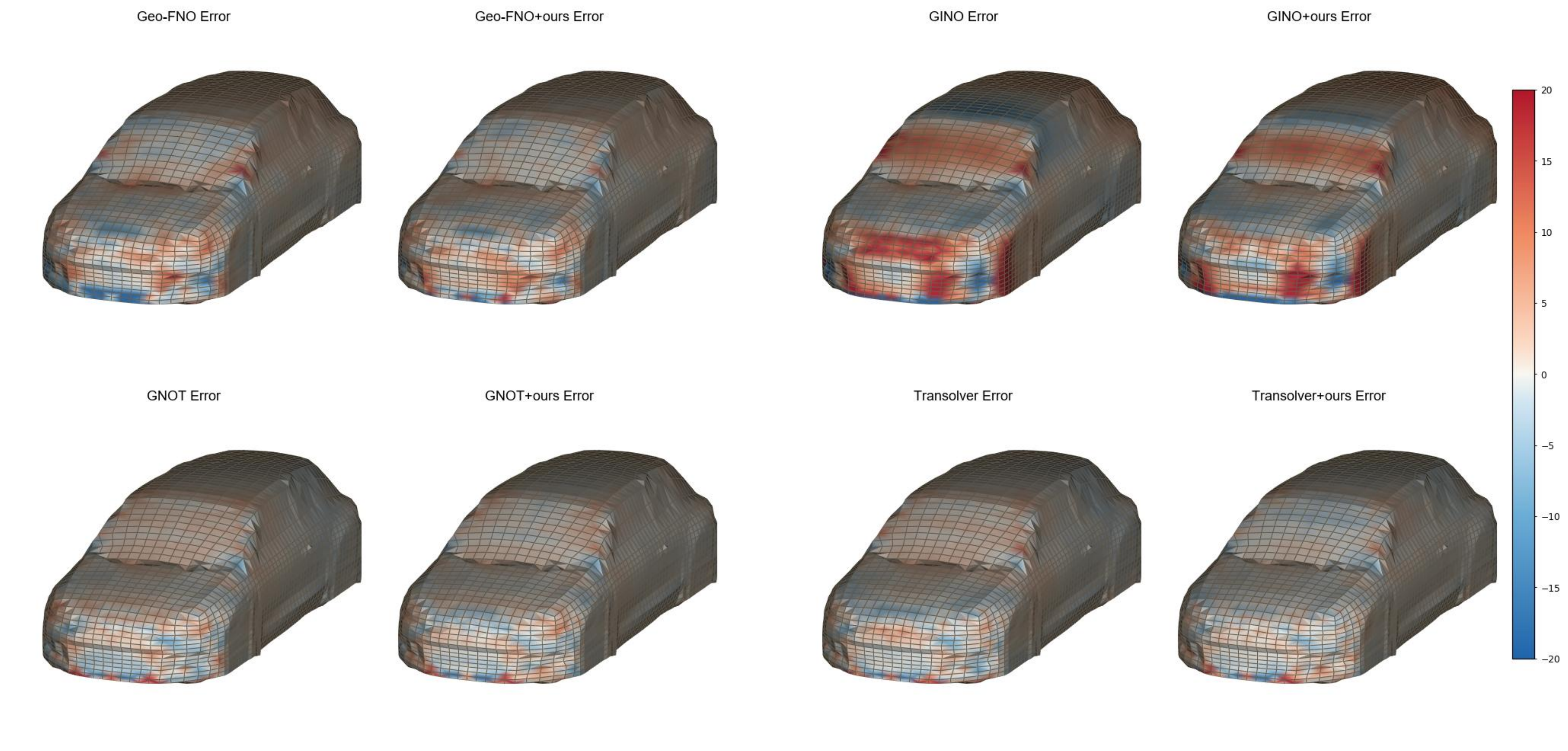}
\caption{\textbf{Surface pressure error visualization on the ShapeNetCar benchmark.}
We visualize the pressure prediction error on the same test car for Geo-FNO, GINO, GNOT, and Transolver, comparing each baseline with its GeoABC-enhanced counterpart under a shared color scale. GeoABC consistently suppresses large spatial error patterns on the vehicle surface, especially around high-curvature and front-body regions, showing that the proposed anisotropic boundary correction improves pressure prediction near solid boundaries across different neural operator backbones.}
    \label{fig:3dcar_pressure_1}
\end{figure}

\begin{figure}[htbp]
    \centering
    \includegraphics[width=1.0\linewidth]{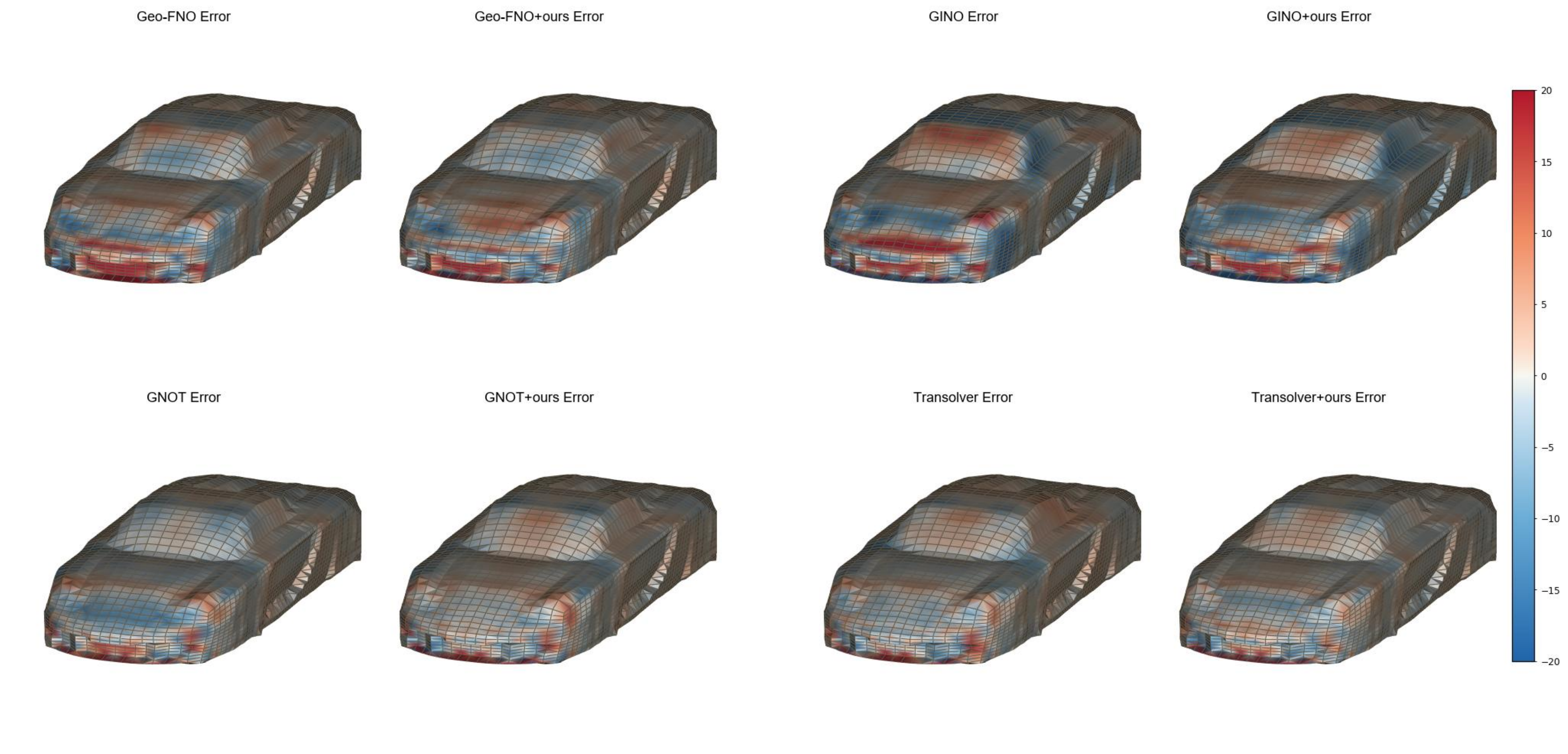}
\caption{\textbf{Surface pressure error visualization on the ShapeNetCar benchmark.}
We visualize the pressure prediction error on the same test car for Geo-FNO, GINO, GNOT, and Transolver, comparing each baseline with its GeoABC-enhanced counterpart under a shared color scale. GeoABC consistently suppresses large spatial error patterns on the vehicle surface, especially around high-curvature and front-body regions, showing that the proposed anisotropic boundary correction improves pressure prediction near solid boundaries across different neural operator backbones.}

    \label{fig:3dcar_pressure_2}
\end{figure}

\FloatBarrier

\section{Training and Implementation Details}
\label{sec:appendix:hparams}

\subsection{Baselines and Common Training Protocol}
\label{sec:appendix:baselines_training}

\paragraph{Baselines.}
We evaluate GeoABC as a plug-in correction module on representative neural operator backbones covering spectral, convolutional, graph/kernel-based, and attention-based operator families.

\begin{itemize}
    \item \textbf{Geo-FNO.}
    Geo-FNO~\citep{li2023geofno} extends Fourier neural operators to irregular geometries by mapping the physical domain to a latent regular grid and applying Fourier layers in the transformed space.
    We use Geo-FNO as the primary spectral operator baseline on the 2D NACA airfoil benchmark.
    For the 3D ShapeNetCar benchmark, we use the corresponding Geo-FNO3D variant with 3D Fourier modes and a latent volumetric grid.

    \item \textbf{CNO.}
    CNO~\citep{raonic2023cno} is a convolutional neural operator that learns resolution-consistent mappings between function spaces through encoder--decoder convolutional blocks.
    It provides a grid-based non-spectral baseline on the 2D NACA benchmark.
    This allows us to test whether GeoABC remains effective beyond Fourier-based operator architectures.

    \item \textbf{GNOT.}
    GNOT~\citep{hao2023gnot} is a transformer-based neural operator designed for irregular meshes and multiple input functions.
    It uses attention-based message passing over query points and latent representations, making it a representative attention-based baseline for both 2D airfoil and 3D car prediction.
    We include GNOT to evaluate whether GeoABC can improve boundary-sensitive prediction in models that already have global interaction capacity.

    \item \textbf{GINO.}
    GINO~\citep{li2023gino} combines local graph/kernel message passing with global Fourier operator layers for large-scale 3D PDE prediction on complex geometries.
    We use GINO on the 3D ShapeNetCar benchmark as a hybrid graph--Fourier baseline, where both surrounding velocity and surface pressure are predicted.

    \item \textbf{Transolver.}
    Transolver~\citep{wu2024transolver} learns physics-aware latent tokens from discretized geometry and applies attention over these learned physical states.
    We evaluate Transolver on both the 2D NACA and 3D ShapeNetCar benchmarks.
    This baseline is particularly important because it is already designed for general geometries and large-scale aerodynamic simulation; improvements on top of Transolver therefore test whether GeoABC contributes beyond a strong geometry-adaptive transformer backbone.
\end{itemize}

\paragraph{Datasets and evaluation setting.}
For the 2D benchmark, we use the body-only NACA airfoil dataset of~\citep{li2023geofno}, with the standard split of $1000$ training, $200$ validation, and $200$ test samples.
The evaluated 2D backbones are Geo-FNO, CNO, GNOT, and Transolver.
For the 3D benchmark, we use the ShapeNetCar external-aerodynamics dataset~\citep{umetani2018learning}, where the model predicts the surrounding velocity components $(u,v,w)$ and pressure $p$ around car geometries.
The evaluated 3D backbones are Geo-FNO3D, GINO, GNOT, and Transolver.
For each backbone, the baseline and its GeoABC-augmented counterpart use the same dataset split, preprocessing procedure, output head, optimizer family, and training horizon.

\paragraph{Common training protocol.}
All baseline and GeoABC models are trained for $50$ epochs unless otherwise stated.
For 2D NACA, checkpoints are selected by validation near-wall rel-$L_2$.
For 3D ShapeNetCar, we follow the benchmark validation protocol and select the validation-best checkpoint according to the corresponding surface-sensitive metric.
The exact backbone-specific architecture, optimizer, schedule, batch size, and GeoABC boundary-correction scale are summarized in \cref{tab:hparams}.

\paragraph{Standalone training.}
GeoABC is trained jointly with the backbone from random initialization.
We do not warm-start from a pretrained baseline, copy baseline weights into the GeoABC run, or share optimizer state between baseline and GeoABC training.
The implementation explicitly disables baseline-to-GeoABC initialization copying.
This strict standalone protocol ensures that each baseline--GeoABC comparison measures the effect of the inserted anisotropic correction module rather than a paired-initialization artifact.

\paragraph{Loss weights and anchor supervision.}
The training loss follows the form defined in \cref{app:loss}.
For 2D NACA, we use $\lambda_b = 2.0$ and $\lambda_a = 0.1$, with the anchor supervised by the ground-truth velocity vector $(u,v)$.
For 3D ShapeNetCar, we use $\lambda_b = 0.5$ and $\lambda_a = 0.02$, matching the scale of the velocity and surface-pressure losses; the anchor space is supervised by the local 3D velocity.
Seed counts for each result block are reported in the corresponding table captions, and results with two or fewer seeds are treated as breadth evidence rather than strong multi-seed claims.

\subsection{Training Objective Details}
\label{app:loss}
GeoABC is trained jointly with the backbone from scratch, without frozen modules or a two-stage protocol.
The overall objective consists of a global prediction loss, a boundary-weighted loss, and an anchor alignment loss:
\begin{equation}
    \mathcal{L}
    =
    \mathcal{L}_{\mathrm{global}}
    +
    \lambda_b \mathcal{L}_{\mathrm{boundary}}
    +
    \lambda_a \mathcal{L}_{\mathrm{anchor}} .
    \label{app:eq:total_loss}
\end{equation}
Specifically,
\begin{align}
    \mathcal{L}_{\mathrm{global}}
    &=
    \frac{1}{|\Omega_G|}
    \sum_{\mathbf{x}\in\Omega_G}
    \left\|
    \hat{\mathbf{u}}_G(\mathbf{x})
    -
    \mathbf{u}_G(\mathbf{x})
    \right\|_2^2,
    \label{eq:global_loss}
    \\
    \mathcal{L}_{\mathrm{boundary}}
    &=
    \frac{1}{Z_G}
    \sum_{\mathbf{x}\in\Omega_G}
    \omega_G(\mathbf{x})
    \left\|
    \hat{\mathbf{u}}_G(\mathbf{x})
    -
    \mathbf{u}_G(\mathbf{x})
    \right\|_2^2,
    \label{eq:boundary_loss}
    \\
    \mathcal{L}_{\mathrm{anchor}}
    &=
    \frac{1}{Z_G}
    \sum_{\mathbf{x}\in\Omega_G}
    \omega_G(\mathbf{x})
    \left\|
    \mathbf{a}(\mathbf{x})
    -
    \Pi\big(\mathbf{u}_G(\mathbf{x})\big)
    \right\|_2^2,
    \label{eq:anchor_loss}
    \\
    Z_G
    &=
    \sum_{\mathbf{x}\in\Omega_G}
    \omega_G(\mathbf{x}).
    \label{eq:normalization_factor}
\end{align}
Here, $\mathcal{L}_{\mathrm{global}}$ preserves the overall field accuracy, while $\mathcal{L}_{\mathrm{boundary}}$ emphasizes reliable near-wall regions through the spatial confidence $\omega_G(\mathbf{x})$.
The anchor loss $\mathcal{L}_{\mathrm{anchor}}$ aligns the anchor representation $\mathbf{a}(\mathbf{x})$ with a physical directional target $\Pi(\mathbf{u}_G(\mathbf{x}))$, such as the velocity vector in the $(u,v,p)$ setting or the pressure-gradient direction in the pressure-only setting.

\subsection{Backbone Configurations}
\label{app:baselineconfig}

\Cref{tab:hparams} summarizes the backbone configurations used in the main cross-backbone experiments.
For each backbone, we report the architecture, optimizer, learning-rate schedule, batch size, and the boundary-correction scale $s$ used by GeoABC.
All other components of the data pipeline, training protocol, and output head are kept identical to the corresponding baseline.

\begin{table}[h]
\caption{\textbf{Backbone configurations for the main experiments.} ``BC'' denotes the boundary-correction scale $s$ used by GeoABC.}
\label{tab:hparams}
\setlength{\tabcolsep}{3.2pt}
\begin{tabular}{p{0.16\linewidth}p{0.30\linewidth}p{0.32\linewidth}cc}
\toprule
Backbone & Architecture & Optimizer and schedule & Batch & BC \\
\midrule
\multicolumn{5}{l}{\textit{2D NACA airfoil}} \\
\quad Geo-FNO 
& width $32$, modes $(24,12)$, $4$ spectral layers 
& AdamW ($5\!\times\!10^{-4}$ / $10^{-4}$), StepLR ($\gamma{=}0.5$, step $25$) 
& $8$ & $0.10$ \\

\quad CNO 
& $4$ encoder / $4$ decoder stages, $4$ ResNet blocks, channel multiplier $16$ 
& AdamW ($10^{-3}$ / $10^{-4}$), StepLR ($\gamma{=}0.5$, step $25$) 
& $8$ & $0.10$ \\

\quad GNOT 
& $n_{\mathrm{embd}}{=}24$, $24$-head cross-attention, $4$ blocks 
& AdamW ($10^{-3}$ / $10^{-4}$), StepLR ($\gamma{=}0.5$, step $25$) 
& $8$ & $0.10$ \\

\quad Transolver 
& $n_{\mathrm{hid}}{=}128$, $L{=}8$, $n_h{=}8$, slice $64$ 
& AdamW ($10^{-3}$ / $10^{-5}$), OneCycleLR 
& $4$ & $0.05$ \\

\midrule
\multicolumn{5}{l}{\textit{3D ShapeNetCar}} \\
\quad Geo-FNO 
& width $48$, modes $(8,6,8)$, grid $(24,12,24)$, $4$ spectral layers 
& AdamW ($5\!\times\!10^{-4}$ / $10^{-4}$), CosineLR 
& $1$ & $0.10$ \\

\quad GINO 
& width $48$, modes $(8,6,8)$, grid $(24,12,24)$, $4$ operator layers 
& AdamW ($5\!\times\!10^{-4}$ / $10^{-4}$), CosineLR 
& $1$ & $0.10$ \\

\quad GNOT 
& width $128$, $4$ heads, $4$ attention blocks, radius $r{=}0.2$ 
& AdamW ($5\!\times\!10^{-4}$ / $10^{-4}$), backbone default 
& $1$ & $0.05$ \\

\quad Transolver 
& $n_{\mathrm{hid}}{=}256$, $L{=}8$, $n_h{=}8$, slice $32$ 
& AdamW ($10^{-3}$ / $10^{-4}$), OneCycleLR 
& $4$ & $0.05$ \\

\bottomrule
\end{tabular}
\end{table}

\subsection{Plug-In Interface}
\label{sec:appendix:plugin}

GeoABC is used as a mid-layer wrapper around an existing neural operator backbone.
The backbone is not redesigned: we only expose one spatially indexed hidden state, apply a geometry-conditioned anisotropic boundary correction with the same layout, and pass the corrected hidden state back to the remaining backbone layers.
Thus, the plug-in contract has three objects: a backbone prefix $\mathcal{F}^{\leq \ell^*}_\theta$, a backbone suffix $\mathcal{F}^{> \ell^*}_\theta$, and the GeoABC injector $\mathcal{I}_\phi$.
The only backbone-specific metadata are the injection layer $\ell^*$, the latent channel dimension $C_{\ell^*}$, the latent layout, and the geometry-cache alignment rule.
All other components---the offline geometry cache, the tangent--normal frame, the three gated correction branches, and the three-term loss---are shared.

\begin{algorithm}[h]
\caption{GeoABC plug-in training for an arbitrary backbone $\mathcal{F}_\theta$}
\label{alg:plugin}
\begin{algorithmic}[1]
\REQUIRE Backbone $\mathcal{F}_\theta$; injection layer $\ell^*$; geometry cache $G$; boundary-correction scale $s$; loss weights $\lambda_b,\lambda_a$.
\STATE Build $G=\{\phi,w,c,\hat{t},\hat{n},\kappa\}$  once from the training geometry and store it with the same node/grid indexing as the data.
\STATE Split the backbone forward pass into a prefix $\mathcal{F}^{\leq \ell^*}_\theta$ and suffix $\mathcal{F}^{> \ell^*}_\theta$; instantiate $\mathcal{I}_\phi$ with latent dimension $C_{\ell^*}$.
\STATE Initialize the baseline and GeoABC runs independently; train $\theta$ and $\phi$ jointly from scratch.
\FOR{each mini-batch $(x,G,y^*)$}
  \STATE $z_{\ell^*}\leftarrow \mathcal{F}^{\leq \ell^*}_\theta(x)$ \hfill\COMMENT{capture a spatially indexed mid-layer latent}
  \STATE $\bar{G}\leftarrow \textsc{AlignGeometry}(G,z_{\ell^*})$ \hfill\COMMENT{resample on grids; direct node lookup on point clouds}
  \STATE $(\mathcal{B}_{\ell^*},\hat{a})\leftarrow \mathcal{I}_\phi(z_{\ell^*},\bar{G})$ \hfill\COMMENT{tangent, normal, and cross-frame boundary correction}
  \STATE $\tilde{z}_{\ell^*}\leftarrow \textsc{ApplyBoundaryCorrection}(z_{\ell^*},\,s\,\mathcal{B}_{\ell^*})$
  \STATE $\hat{y}\leftarrow \mathcal{F}^{> \ell^*}_\theta(\tilde{z}_{\ell^*})$
  \STATE $\mathcal{L}\leftarrow \mathcal{L}_{\mathrm{global}}(\hat{y},y^*)+\lambda_b\mathcal{L}_{\mathrm{boundary}}(\hat{y},y^*;\bar{G})+\lambda_a\mathcal{L}_{\mathrm{anchor}}(\hat{a},a^*;\bar{G})$
  \STATE Update $(\theta,\phi)$ with the backbone's original optimizer and schedule.
\ENDFOR
\end{algorithmic}
\end{algorithm}

\noindent
\textbf{What ``plug-in'' means.}
GeoABC does not require a new solver head, a frozen pretrained backbone, or a two-stage adapter fit.
For each backbone, we keep the original input pipeline, output head, optimizer, and training schedule.
The implementation change is a forward-pass interception: run the original model until $\ell^*$, apply the GeoABC boundary correction inside the trusted boundary band, then resume the original model.
Because $\mathcal{B}_{\ell^*}$ has exactly the same layout as the intercepted latent, the suffix of the backbone is unaware that a correction module has been inserted.
When $w(x)c(x)=0$, the boundary correction is inactive and the wrapper locally reduces to the original backbone.

\subsection{Backbone-Specific Adaptation}
\label{sec:appendix:adapt}

Adapting GeoABC to different backbones only changes how the common plug-in interface is realized.
The GeoABC injector itself is shared across all experiments.
For each backbone, we intercept an intermediate representation that still preserves spatial correspondence with grid cells, mesh nodes, or query points, convert it into a pointwise latent format when necessary, apply the same geometry-conditioned boundary correction, and then return the corrected representation to the original backbone layout.

\begin{itemize}
    \item \textbf{Geo-FNO.}
    Geo-FNO exposes channel-first latent tensors from its spectral operator stack.
    We intercept the intermediate feature after \texttt{conv2} and before the following spectral block.
    The latent tensor is temporarily permuted into a pointwise representation so that each spatial location has one feature vector.
    The geometry cache is aligned to the current latent resolution by interpolation.
    Scalar geometric quantities, including the signed distance, narrow-band window, confidence score, and curvature, are interpolated directly.
    Vector quantities, including tangent and normal directions, are interpolated and then re-normalized before being passed to GeoABC.
    After the pointwise correction is applied, the corrected latent is reshaped back to the original channel-first tensor layout and forwarded through the remaining Geo-FNO layers.
    This adaptation keeps the Fourier backbone unchanged and only inserts GeoABC as a mid-layer boundary correction module.

    \item \textbf{CNO.}
    CNO uses an encoder--decoder convolutional structure with intermediate feature maps at multiple resolutions.
    We insert GeoABC after the first encoder downsampling ResNet stage, where the representation has already aggregated local physical context while still retaining a clear spatial correspondence to the computational grid.
    Since the intercepted CNO feature map has a lower resolution than the original geometry cache, all scalar geometric descriptors are downsampled to the encoder resolution.
    The tangent and normal fields are also downsampled, followed by normalization to preserve a valid local boundary frame.
    The corrected feature map is then restored to the original CNO tensor layout and passed to the remaining encoder--decoder stages.
    In this way, GeoABC corrects the boundary-related latent representation before it is further propagated through the multiscale convolutional operator.

    \item \textbf{GNOT.}
    GNOT maintains point-indexed latent representations at its attention blocks.
    We insert GeoABC after the first cross-attention block, where each mesh node or query point already has a corresponding latent vector.
    Because the latent representation and the geometry cache share the same node order, no spatial interpolation is required.
    GeoABC directly looks up the local geometric descriptor for each node, applies the tangent--normal boundary correction in the pointwise latent space, and returns the corrected tokens to the remaining attention blocks.
    This adaptation tests whether explicit anisotropic boundary correction can still improve a backbone that already uses global attention-based interactions.

    \item \textbf{GINO.}
    GINO combines local graph/kernel message passing with global Fourier-style operator layers.
    We insert GeoABC at a midpoint operator layer, where the representation has accumulated both local geometric interactions and global operator context.
    When the intercepted latent is stored as a grid-like tensor, we use the same wrapper as in Geo-FNO: the tensor is converted to a pointwise format, the geometry cache is aligned to the latent resolution, and the corrected feature is reshaped back afterward.
    When the latent is already associated with sampled query points, GeoABC is attached by direct pointwise geometry lookup.
    The corrected representation is then passed to the remaining GINO operator layers and the original output head.
    This preserves the original GINO computation while adding a localized anisotropic correction near solid boundaries.

    \item \textbf{Transolver.}
    Transolver represents the physical domain through latent features associated with query points and learned physical states.
    We insert GeoABC after transformer block $\ell=3$ of the $L=8$ stack.
    At this stage, the representation has already interacted through the physics-aware attention mechanism, but it still provides one feature vector for each query point after the block output.
    GeoABC is applied by direct geometry lookup at the corresponding query positions.
    The corrected latent vectors are then fed into the remaining transformer blocks and the original output head.
    This adaptation is important because Transolver is already a strong geometry-adaptive transformer backbone; improvements after inserting GeoABC therefore indicate that the proposed correction provides boundary-specific structure beyond generic geometry-aware attention.

    \item \textbf{Three-dimensional extension.}
    For volumetric or surface-based backbones, the same adaptation principle is used.
    The only change is that the anchor projection maps the intercepted latent to a three-dimensional anchor space rather than a two-dimensional one.
    The local boundary frame is constructed from the surface normal and two tangent directions spanning the tangent plane.
    Grid-based tensors use trilinear alignment of the geometry cache, while point-cloud or token-based tensors use direct geometry lookup by node or query-point identity.
    The rest of the GeoABC injector, including geometric gating, direction-aware branches, correction lifting, and latent write-back, remains unchanged.
\end{itemize}

\paragraph{Layer selection.}
Across all backbones, GeoABC is inserted at an intermediate layer rather than at the input or output stage.
An input-level insertion would reduce the method to ordinary boundary-feature augmentation before sufficient physical context is formed.
An output-level insertion would behave like a late post-correction with limited opportunity to interact with subsequent operator layers.
Therefore, the chosen layers balance two requirements: the representation should already encode global physical context, while still preserving a clear spatial correspondence to grid cells, mesh nodes, or query points.
For each backbone, the selected insertion location is fixed across all seeds and variants.

\subsection{Injector Internals}
\label{sec:appendix:injector}

The GeoABC injector is implemented as a lightweight pointwise correction module shared across all backbones.
Its input is the intercepted mid-layer latent representation together with the offline geometry cache.
Backbone wrappers only convert the latent layout when necessary; the internal injector architecture remains the same.

The injector consists of four components.
First, an anchor projection maps the backbone latent into a low-dimensional physical anchor space.
The anchor dimension is two for 2d tasks and three for 3d tasks, matching the dimensionality of the local velocity used for anchor supervision.
Second, an auxiliary projection extracts a compact context vector from the same latent representation.
This context is shared by all correction branches and provides local physical information from the backbone feature.
Third, a geometry modulator takes scalar geometric descriptors from the offline cache, including signed distance, boundary window, confidence, and curvature, and produces directional gates for the correction branches.
Finally, three pointwise branches generate tangent, normal, and tangent--normal coupling corrections in the local boundary frame.

The correction is activated only inside the trusted near-wall band.
In practice, the boundary window and confidence score jointly control the spatial support of the correction, while the boundary-correction scale controls its overall magnitude.
Therefore, points far from the solid boundary keep the original backbone representation unchanged, and points inside the near-wall band receive geometry-conditioned anisotropic correction.

For 2d airfoil experiments, the local frame is defined by the tangent and normal directions of the closest boundary point.
For 3d car experiments, the frame is defined by the surface normal and two tangent directions spanning the local tangent plane.
Apart from this change in frame dimension, the same injector design is used across 2d and 3d settings.

\begin{table}[t]
\centering
\small
\caption{\textbf{Parameter overhead of the GeoABC injector.}
GeoABC introduces only a small number of additional parameters relative to the original backbone.}
\label{tab:injector_params}
\setlength{\tabcolsep}{8pt}
\begin{tabular}{lcccc}
\toprule
Backbone & Baseline params & GeoABC params & Added params & Relative increase \\
\midrule
Geo-FNO
& $2{,}368{,}291$
& $2{,}397{,}179$
& $28{,}888$
& $1.22\%$ \\
CNO
& $2{,}007{,}211$
& $2{,}034{,}723$
& $27{,}512$
& $1.37\%$ \\
\bottomrule
\end{tabular}
\end{table}

Overall, the injector adds less than $2\%$ parameters for the representative backbones reported in \cref{tab:injector_params}.
This confirms that the improvement of GeoABC does not come from substantially increasing model capacity, but from introducing geometry-conditioned anisotropic correction into the intermediate representation.

\subsection{Seed Protocol and Reporting}
\label{sec:appendix:seed}

All experiments follow a strict standalone protocol: each baseline and each GeoABC run is trained from an independent random initialization, with no shared optimizer state and no warm-start from a pretrained backbone.
For the experiments where variance is reported, we use three random seeds, $\{0,1,2\}$, and report the results as mean $\pm$ standard deviation.

\subsection{Code Availability and Computing Resources}
The anonymized source code is available at
\url{https://anonymous.4open.science/r/GeoABC_code/}.
The repository provides the implementation of GeoABC and instructions for reproducing the main experimental results.

All experiments were conducted on a workstation equipped with four NVIDIA GeForce RTX 2080 Ti GPUs. The reported training runs were executed on this local GPU machine using the training configurations described in Appendix~\ref{app:baselineconfig}.

\section{Metric Computation}
\label{sec:appendix:metrics}

All metrics are evaluated on the held-out test set with the validation-best checkpoint.
Let $\hat{y}$ denote the model prediction and $y^*$ the ground truth over $N$ mesh points.

\paragraph{Global and near-wall rel-$L_2$.}
\begin{equation*}
    \mathrm{rel}\text{-}L_2 \;=\; \frac{1}{N_\mathrm{test}} \sum_{i=1}^{N_\mathrm{test}} \frac{\|\hat{y}^{(i)} - y^{*(i)}\|_2}{\|y^{*(i)}\|_2}.
\end{equation*}
For multi-channel outputs, the per-channel ratio is averaged.
\emph{Global} rel-$L_2$ uses all $N$ mesh points; \emph{near-wall} rel-$L_2$ uses only the points with $\mathrm{window}(x) > 0$, corresponding to the top ${\sim}5.4\%$ $\phi$-percentile band around the body surface.

\paragraph{$C_p$ MAE and wall-pressure MAE.}
The surface pressure coefficient is $C_p = (p - p_\infty) / (\frac{1}{2}\rho_\infty u_\infty^2)$.
Both metrics are computed on the $121$-point body arc:
\begin{equation*}
    C_p\text{-MAE} = \frac{1}{|\partial\Omega|}\sum_{j \in \partial\Omega} |\hat{C}_p^{(j)} - C_p^{*(j)}|, \qquad
    p\text{-MAE} = \frac{1}{|\partial\Omega|}\sum_{j \in \partial\Omega} |\hat{p}^{(j)} - p^{*(j)}|.
\end{equation*}

\paragraph{3D surface-velocity rel-$L_2$.}
We keep the benchmark-style name ``surface-velocity'' for consistency with the
3D result tables, but the metric is evaluated on near-surface velocity points
rather than directly on the surface mesh. In TransolverCar, velocity labels are
defined on off-surface CFD volume points, while pressure labels are defined on
wetted-surface points. Let $\mathcal{V}$ denote the non-surface velocity points
and let $d_i$ be the denormalized SDF value of point $i$. We use the
SDF-defined near-surface mask
\[
\mathcal{M}_{\mathrm{surf\mbox{-}vel}}
=
\left\{
i\in\mathcal{V}:
\exp\!\left[-\left(d_i/0.2\right)^2\right] > 0.1
\right\}.
\]
The reported 3D surface-velocity rel-$L_2$ is then
\[
\frac{
\left(
\sum_{i\in\mathcal{M}_{\mathrm{surf\mbox{-}vel}}}
\|\hat{\mathbf{u}}_i-\mathbf{u}_i\|_2^2
\right)^{1/2}
}{
\left(
\sum_{i\in\mathcal{M}_{\mathrm{surf\mbox{-}vel}}}
\|\mathbf{u}_i\|_2^2
\right)^{1/2}
}.
\]
Thus, throughout the paper, ``surface-velocity'' refers to this
surface-proximal velocity-band metric, whereas surface-pressure metrics are
computed directly on wetted-surface points.

\section{Limitations}
This work focuses on near-wall prediction in steady-state aerodynamic surrogate modeling.
GeoABC shows consistent improvements on both 2D airfoil and 3D car geometries, demonstrating the effectiveness of using local tangential--normal structures for boundary correction.
At the same time, our current experiments are centered on steady external-flow settings.
An important future direction is to extend this idea to unsteady flows, broader operating-condition ranges, and multiphysics coupled systems.

In addition, GeoABC constructs local geometric descriptors from an offline geometry cache, including the tangent--normal frame, curvature, near-wall window, and confidence score.
This design allows the method to explicitly leverage boundary geometry that is naturally available in CFD tasks.
Future work may further investigate more automated and robust geometry preprocessing pipelines, making the framework easier to apply to larger-scale and more complex industrial geometries.

Finally, this paper introduces geometric structural priors from a data-driven neural operator perspective to improve boundary-sensitive prediction.
Future work can further explore the combination of GeoABC with explicit physical constraints, such as PDE residuals, boundary-condition regularization, or conservation-preserving structures.
Such integration may further improve the reliability of neural aerodynamic surrogates under different operating conditions, geometry distributions, and high-fidelity engineering simulation settings.

\section{Broader Impacts}
\label{app:broader_impacts}

This work studies a geometry-conditioned anisotropic boundary correction module for neural operator surrogates in steady-state aerodynamic simulation. Our primary motivation is scientific and engineering: improving the near-wall accuracy of learned CFD surrogates so that they can more reliably support shape design, rapid screening, and engineering iteration. We discuss below both potential positive societal impacts and possible negative considerations of this line of research.

\paragraph{Potential positive impacts.}
Classical CFD solvers, although accurate, are computationally expensive and often become a bottleneck in design optimization, parameter sweeps, and interactive engineering iteration. Neural operator surrogates can reduce this cost by orders of magnitude, and improving their boundary region accuracy is critical for engineering metrics such as the surface pressure coefficient $C_p$ and the drag coefficient $C_D$. By making such surrogates more reliable on quantities that directly affect lift, drag, and aerodynamic loading, GeoABC can contribute to (i) faster and more energy-efficient aerodynamic design cycles for vehicles, aircraft, and other engineering shapes; (ii) reduced computational and energy footprint of large-scale design exploration that would otherwise rely on extensive high-fidelity CFD runs; and (iii) easier access to aerodynamic analysis for research groups and small engineering teams that lack high-end CFD infrastructure. In the long term, more accurate surrogate models can also support broader sustainability goals, for example, by enabling more thorough exploration of low-drag vehicle shapes or more efficient wing geometries.

\paragraph{Potential negative impacts and risks.}
The methodology proposed in this paper is foundational research on neural PDE surrogates and does not generate content, identify individuals, or make decisions about people. We do not foresee a direct path to malicious applications such as disinformation, surveillance, or biased decision-making. Nevertheless, two indirect risks deserve mention. First, like any surrogate model, GeoABC may produce predictions that are accurate on the training
distribution but less reliable on out-of-distribution geometries or operating conditions. If such surrogates are deployed in safety-critical engineering pipelines without adequate validation against high-fidelity solvers or experiments, prediction errors could propagate into design decisions. Second, the same acceleration that benefits civilian aerodynamic design could in principle be used in dual-use engineering contexts. We note, however, that the techniques studied here are general-purpose CFD surrogate improvements and do not provide capabilities specific to any particular sensitive application.

\paragraph{Mitigation and responsible use.}
We encourage practitioners who build on this work to (i) treat neural operator surrogates as accelerators rather than replacements for high-fidelity CFD, especially for safety- or certification-critical decisions; (ii) systematically report out-of-distribution behavior and near-wall reliability when deploying such models in engineering pipelines; and (iii) combine data-driven surrogates with physics-based checks, such as PDE residual evaluation or conservation diagnostics, to flag potentially unreliable predictions. The methodology, datasets, and evaluation protocols used in this paper are all standard and publicly available, which allows independent verification of the reported results.

Overall, we believe that the broader impact of this work is primarily positive: it advances the reliability of neural CFD surrogates in a principled, geometry-aware manner, while the residual risks are common to the field of learned PDE solvers and can be mitigated through careful validation and responsible deployment.

\end{document}